%% file: main.tex
\pdfoutput=1
\documentclass[aps,prx,reprint,superscriptaddress,bibnotes,longbibliography,floatfix]{revtex4-2}
\usepackage{amsmath,amssymb,amsthm,graphicx,booktabs,array,longtable,xcolor,placeins}
\usepackage[colorlinks=true,allcolors=blue!45!black]{hyperref}
\newtheorem{theorem}{Theorem}

\newtheorem{lemma}[theorem]{Lemma}
\theoremstyle{definition}
\newcommand{\F}{\mathbb F_2}
\newcommand{\CSS}{\operatorname{CSS}}
\newcommand{\wt}{\operatorname{wt}}

\newcommand{\minover}[1]{\min_{\substack{#1}}}%
\newcommand{\eqindent}{\hspace*{0.8em}}%
\newcommand{\CNOT}{\mathrm{CNOT}}
\newcommand{\CCZ}{\mathrm{CCZ}}
\newcommand{\CS}{\mathrm{CS}}
\newcommand{\Qh}{Q_{37}}

\newcommand{\ket}[1]{|#1\rangle}
\newcommand{\code}[5]{[\![#1,#2,#3,#4,#5]\!]}
\newcommand{\qcode}[3]{[\![#1,#2,#3]\!]}
\newcommand{\Wt}{\mathrm{Wt}}%
\usepackage{totcount}
\newcommand{\Dpart}{d_{\mathrm{part}}}%
\newcommand{\tpart}{t_{\mathrm{part}}}%

\newtotcounter{ftcodes}\newtotcounter{pscodes}
\DeclareRobustCommand{\NFT}{\total{ftcodes}}
\DeclareRobustCommand{\NPS}{\total{pscodes}}
\DeclareRobustCommand{\NDB}{\ifnum\totvalue{ftcodes}<0 $??$\else\ifnum\totvalue{pscodes}<0 $??$\else\number\numexpr\totvalue{ftcodes}+\totvalue{pscodes}\relax\fi\fi}

\begin{document}
\title{Qubit CSS codes with Universal Transversal Gates}

\author{Arpit Dua}
\email{adua@quera.com}
\affiliation{QuEra Computing Inc., 1380 Soldiers Field Road, Boston, MA 02135, USA}

\author{Adam Holmes}
\affiliation{NVIDIA, Santa Clara, CA, USA}

\author{Michael J. Gullans}
\affiliation{QuEra Computing Inc., 1380 Soldiers Field Road, Boston, MA 02135, USA}

\author{Victor V. Albert}
\email{vva@umd.edu}
\affiliation{Joint Center for Quantum Information and Computer Science,
NIST/University of Maryland, College Park, MD 20742, USA}

\date{\today}
\begin{abstract}

The Eastin--Knill theorem rules out a universal set of transversal logical gates on any code that detects single-qubit errors. The theorem constrains one partition of the qubits at a time. We show that depth-one layers of few-qubit gates, each respecting its own partition, can together generate a universal logical gate set on an error-correcting code.
Each logical gate generator is then a single depth-one layer followed by error correction of the same stabilizer code, with no magic-state distillation, code switching, gauge fixing, lattice surgery, or other gadgets. 
We show that such codes can be found by applying an entangling circuit to a lower-distance ``seed'' code admitting a depth-one non-Clifford gate and an appropriate ``kernel'' code.
We construct \NFT{} codes with depth-one generating sets that tolerate at least one faulty gate per layer, including a $\qcode{37}{1}{7}$ and a $\qcode{59}{1}{9}$ code, both of check weight at most eight.

\end{abstract}
\maketitle

\section{Introduction}
\label{sec:intro}
A fault-tolerant quantum computer stores its data in quantum error-correcting codes and must compute without letting errors spread faster than they are corrected. Transversal gates, naively, offer a simple way to do this: each physical gate touches one qubit per code block, so a single fault creates at most one error per block~\cite{gottesman2010introduction}. If a universal gate set had this form on one code, the overhead cost of a computation would be similar to that of a fault-tolerant memory.
Unfortunately, the Eastin--Knill theorem rules out a universal set of transversal gates for qubit codes that detect single-qubit errors~\cite{eastin2009,zeng2011,chen2008subsystem}, and further limitations of such gates have been proven~\cite{jochymoconnor2018disjointness,chakraborty2026nogo,koh2026achieving,guyot2026addressability}.

As a result of these limitations, fault-tolerant compilers typically incorporate non-unitary gadgets requiring measurements or resets beyond those used for syndrome extraction of a fixed code~\cite{campbell2017roads}. 
Such gadgets include magic-state distillation, cultivation, or injection~\cite{bravyi2005,litinski2019distillation,gidney2024cultivation,li2015magic};
lattice surgery~\cite{horsman2012surface,cohen2022low};
code switching, code deformation, or gauge fixing~\cite{paetznick2013,anderson2014,bombin2015,butt2025measurement,veroni2025universal};
pieceable fault tolerance~\cite{yoder2016,yoder2017universal};
and flag-based gadgets~\cite{anker2025}.
Although their overhead has been studied extensively, magic-state factories can still take up a substantial fraction of an algorithm's spacetime volume~\cite{fowler2012surface,litinski2019game,gidney2021rsa,beverland2022}.
Magic-state cultivation has brought the cost of a $T$ state close to that of a lattice-surgery CNOT of equal reliability~\cite{gidney2024cultivation}, but it still requires post-selection and teleportation.
Other schemes avoid additional measurements, but use deeper non-transversal circuits that can include long-range qubit permutations~\cite{jochymoconnor2014,chamberland2016thresholds,jochymoconnor2019fault,lavasani2019universal,zhu2020universal}.

\begin{figure*}[t]
\begin{minipage}[b]{0.455\textwidth}
\vbox to 0pt{\hbox{\textbf{(a)}}\vss}\nointerlineskip%
\includegraphics[width=\linewidth]{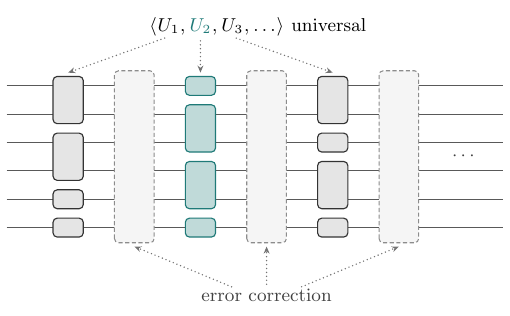}
\end{minipage}\hfill
\begin{minipage}[b]{0.52\textwidth}
\vbox to 0pt{\hbox{\textbf{(b)}}\vss}\nointerlineskip
\includegraphics[width=\linewidth]{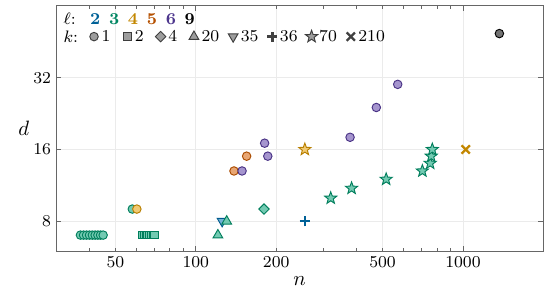}
\end{minipage}
\caption{(a) A logical gate compiled from depth-one layers, each made of disjoint gates on at most $\ell$ qubits of one partition of the qubits (horizontal lines), with error correction between layers. Gray layers are Clifford, and the teal layer $U_2$ is non-Clifford. The layers of any one partition generate a finite logical group whenever every part has fewer than $d$ qubits, but layers on several partitions can generate a universal one. (b) Distance against length for the \NFT{} codes of Table~\ref{tab:ft-codes}, whose $\ell$-fold transversal generators satisfy $\ell\le\lfloor(d-1)/2\rfloor$, so that at least one faulty gate per layer is correctable. The color gives $\ell$ and the symbol the number $k$ of logical qubits.}
\label{fig:concept}
\end{figure*}

The Eastin--Knill no-go result relies on a-priori fixing a qubit partition.
However, it does not constrain the logical group generated by few-body logical layers on \textit{different} partitions [Fig.~\ref{fig:concept}(a)].
Indeed, Eastin and Knill themselves remarked that a code with a universal set of encoded gates, each transversal in isolation, would be extremely useful~\cite{eastin2009}.
Recent works have explored this idea, constructing error-correcting codes that admit ``fold-transversal'' and related gates on different partitions to achieve the full logical Clifford group~\cite{moussa2016,breuckmann2024,gong2024computation,malcolm2025computing,berthusen2025concatenated,tansuwannont2026full,holmes2026qlc,davenport2026generalized,albert2026beyond,matsuura2026constant}.
Other works study qubit codes that admit specific non-Clifford gates in low depth and/or width (see Refs.~\cite{knill1996,bravyi2005,bombin2007topological,bravyi2012,bombin2015,kubica2015unfolding,haah2018,vasmer2019surface,koutsioumpas2022smallest,jain2024tcodes,lin2024transversal,zhu2025nonclifford,golowich2025asymptotically,nguyen2025good,wills2025constant,golowich2025quantum,he2025addressable,golowich2025near,breuckmann2026cups,menon2026magic,wills2026improved,golowich2026cup,wills2026classification}, to name a few).
We take this idea to its natural conclusion and construct codes that have the best of both worlds.

We construct codes admitting a universal set of gates generated by depth-one few-body layers, explain the origin of their non-Clifford generators, and simulate one of the smallest fault-tolerant examples under circuit-level noise.
Our results imply that the bulk of a fault-tolerant quantum computation is possible on a \textit{fixed} stabilizer code using only two independent gadgets interleaved with each other, namely, depth-one layers of few-qubit gates and error correction.

We call a depth-one layer of disjoint gates on at most $\ell$ qubits $\ell$-fold transversal.
The usual notion of a transversal layer is the $\ell=1$ case, where there is only one possible partition. The $\ell=2$ case is thus the minimal extension one has to make to evade the Eastin--Knill theorem because that allows us to have multiple partitions.
We construct \NDB{} error-correcting codes that admit a universal set of logical gates using $\ell$-fold transversal layers on different partitions with $\ell$ ranging from 2 to 10.

The ability to utilize different partitions comes at a price.
An error on $s$ qubits before an $\ell$-fold layer can spread to $\ell s$ qubits by the end of the layer.
One faulty gate per layer is correctable if the distance $d\ge2\ell+1$ (Sec.~\ref{sec:nonclifford}).
We show that \NFT{} of our codes meet this condition for every generator, Clifford and non-Clifford alike [Fig.~\ref{fig:concept}(b)].
The remaining \NPS{} codes are only guaranteed to detect a faulty gate, $d\ge\ell+1$, and need post-selection to be fault tolerant~\cite{knill2005quantum,aliferis2008accuracy,gottesman2016small}.
They serve as toy examples and include Hamming subcodes~\cite{qecdb,grassl2013leveraging,albert2026beyond}, the tesseract code~\cite{hastings2017small,haah2017magic,delfosse2020short,prabhu2024distance,reichardt2024,paetznick2026improved} (Fig.~\ref{fig:tesseract}), and the Carbon code~\cite{knill2005quantum,paetznick2024demonstration,paetznick2026improved}.
An \(\ell\)-fold transversal layer with \(\ell \geq d\) is always possible, e.g., using a minimum-weight logical Pauli as a rotation generator.

The depth-one layers of our ``depth-one universal'' codes consist of a set of generators for the full logical Clifford group, similar to codes in the previous works~\cite{moussa2016,malcolm2025computing,tansuwannont2026full,holmes2026qlc,davenport2026generalized,albert2026beyond}, together with a non-Clifford generator.
We trace the non-Clifford generator of each code back to a pair of shorter codes --- a ``seed'' code with a depth-one non-Clifford gate and an ancillary ``kernel'' code --- that are related to the universal code via an entangling CNOT circuit.
Our identified combinations of kernels, seeds, and CNOT circuits achieve three desiderata simultaneously: they enlarge the distance of the seed code, maintain a nontrivial logical non-Clifford action, and ensure that all generators, Clifford and non-Clifford, are \(\ell\)-fold transversal.

The CNOT circuit spreads the seed's non-Clifford gate over small groups of qubits, and the resulting gate acts on the parity of each group with its logical action unchanged.
We find that the seed sets two limits on this parity layer. 
First, the distance of the code is at most the size of the largest group times the $Z$ distance of the seed, the weight of its lightest $Z$-type logical operator (Sec.~\ref{sec:distance}). 
Second, the parity layer corrects no more faulty gates than that seed's $Z$ distance allows, irrespective of how much the CNOT circuit increases the distance of the seed (Sec.~\ref{sec:nonclifford}). 
A seed of distance three therefore gives a parity layer that corrects at most one faulty gate.

Sections~\ref{sec:tesseract} and~\ref{sec:carbon} show that the tesseract code and the Carbon code are depth-one universal, the latter as a first instance of the seed--kernel construction.
Section~\ref{sec:construction} presents the general construction, the exact distance formulas, the non-Clifford and Clifford layers with their partition distances, and the catalog of \NDB{} codes.
Section~\ref{sec:example} describes a $\qcode{37}{1}{7}$ code on the 15-qubit Reed--Muller seed, one of the smallest that correct a faulty native gate in every generating layer, with its gates, a distance-preserving extraction circuit, and its logical error rates under circuit-level noise.
Section~\ref{sec:conclusion} discusses the open problem of larger partition distances and the questions of architecture and resources that these codes raise.
Appendix~\ref{app:theorems} proves the results of Sec.~\ref{sec:construction}, Appendix~\ref{app:37} specifies the $\qcode{37}{1}{7}$ code and proves its properties, and Appendix~\ref{app:resources} gives its extraction circuit, the noise model and the simulations.

\section{Example: the tesseract is depth-one universal}\label{sec:tesseract}

The tesseract code is a CSS code on 16 qubits that encodes six logical qubits with distance four, i.e., a $\qcode{16}{6}{4}$ code, listed as \texttt{ml:e1c9cb94} in Table~\ref{tab:ps-codes}.
It is the smallest error-correcting member of a Reed--Muller family whose full logical Clifford group is generated by two-fold transversal gates~\cite{tansuwannont2026full}.
We show that the tesseract also admits a two-fold transversal non-Clifford generator.
\begin{figure}[!t]
\centering\includegraphics[width=\columnwidth]{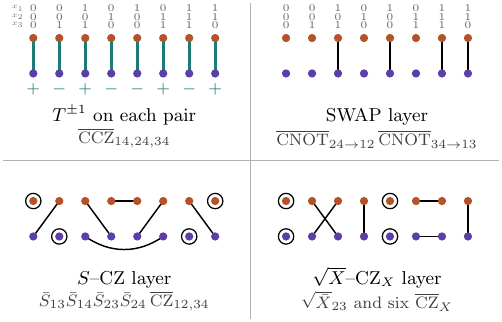}
\caption{
The four generators of a universal two-fold transversal gate set for the tesseract code.
The code's 16 qubits are arranged in two rows, with each corresponding to one of the two cubes of a tesseract and labeled by bit strings $x=x_1x_2x_3x_4$.
The non-Clifford layer is shown in the upper left, and the three Clifford generators are shown in the other panels.
In each panel, lines mark the support of two-qubit gates, and circles mark qubits that receive a single-qubit gate.
Physical and logical actions are described.
For example, the non-Clifford layer (teal) applies $T$ ($+$) or $T^\dagger$ ($-$) to the parity of each pair, yielding logical CCZ gates. 
The diagonal $S$--CZ layer applies $S$ and CZ gates; the $\sqrt{X}$--$\mathrm{CZ}_X$ layer applies $X$-basis versions of those gates.
}
\label{fig:tesseract}
\end{figure}
The code detects but cannot correct all two-qubit errors, so the depth-one layers are only fault tolerant with post-selection.
\begin{figure*}[t]
\centering\includegraphics[width=0.8\textwidth]{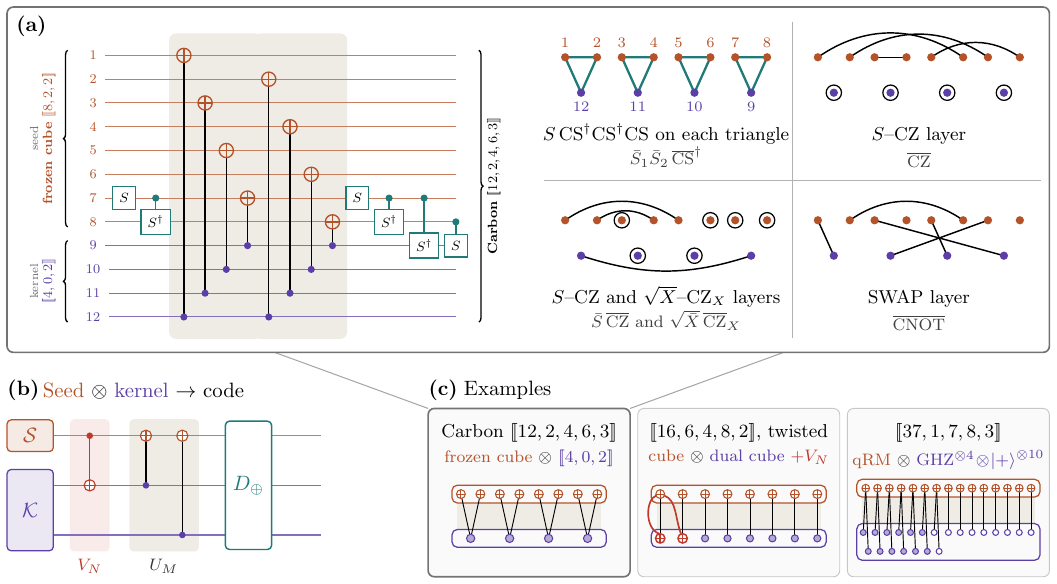}
\caption{Codes obtained by coupling a seed (orange) to a kernel (purple) with CNOTs (shaded). (a) The Carbon code of Sec.~\ref{sec:carbon}, opened from its card in (c). Left: its coupling circuit, with the seed's non-Clifford pair layer (teal) carried to the triangle $\{7,8,9\}$. Right: its depth-one layers on four partitions, drawn as in Fig.~\ref{fig:tesseract}, with the non-Clifford layer on the triangles (teal). (b) General framework: The CNOTs $U_M$ write each seed bit $s_i$ into the parity of the qubits of $F_i$, so a diagonal seed gate $D$ becomes the same gate $D_\oplus$ on these parities, for every kernel $\mathcal K$ and every twist $V_N$ (red), whose CNOTs run from the seed qubit into the kernel before $U_M$. (c) Examples, each named as seed $\otimes$ kernel. Each line is a coupling CNOT from a kernel qubit onto a seed qubit ($\oplus$). Open dots are kernel qubits in $\ket+$, filled dots the other kernel qubits, and red marks the CNOTs of a twist $V_N$. qRM denotes a quantum Reed--Muller seed; for the $\code{37}{1}{7}{8}{3}$ code, where $\Wt=8$ is the largest weight of its stabilizer checks and $\ell=3$ the width of its physical gates, the seed code is the $\qcode{15}{1}{3}$ code.}
\label{fig:coupling}
\end{figure*}

The qubits of the tesseract code are the 16 vertices $x=x_1x_2x_3x_4\in\F^4$ of a tesseract, the four-dimensional cube.
Figure~\ref{fig:tesseract} draws them as two rows of eight, the cube $x_4=0$ above the cube $x_4=1$, so that the eight joining pairs $\{x,x+e_4\}$, where $e_4=0001$ flips the last bit, connect the two rows.
The stabilizer group of this self-dual code is generated by ten weight-eight checks, supported on the eight vertices with $x_i=1$ for each $i$ and on the eight even-weight vertices,
\begin{equation}
 X(c),\ Z(c),\qquad c\in\Bigl\{\{x_i=1\}_{i=1}^{4},\ \bigl\{\textstyle\sum_ix_i=0\bigr\}\Bigr\}~.
 \label{eq:tesschecks}
\end{equation}
Here $X(c)=\prod_{q\in c}X_q$, and likewise for $Z(c)$.
The six logical qubits are labeled by the six two-element subsets $I\subset\{1,2,3,4\}$.
Every weight-four logical operator is supported on a plane $\{x,x+u,x+v,x+u+v\}$, where $u,v\in\F^4$ are distinct and nonzero.
Conversely, $X$ and $Z$ on each of the 140 planes are logical operators.
Among the planes are the 24 squares, on which two coordinates are fixed, and the operators on squares
\begin{align}
 \bar X_I&=X(\{x:x_i=0\text{ for }i\in I\}),\nonumber\\
 \bar Z_I&=Z(\{x:x_j=0\text{ for }j\notin I\})~,
 \label{eq:tesslogicals}
\end{align}
form a symplectic basis, i.e., $\bar X_I$ anticommutes with $\bar Z_J$ exactly when $I=J$.

Figure~\ref{fig:tesseract} draws the four generators of a universal gate set and their action on the tesseract's qubits.
Each generator is a two-fold transversal layer: it applies two-qubit gates to the disjoint pairs drawn as lines and single-qubit gates to the ringed qubits.
The four layers act on four different partitions of the qubits.
The three Clifford layers shown provide a minimal generating set for the full logical Clifford group $\mathrm{Sp}(12,\F)$ (modulo Paulis and phases).
We focus on the non-Clifford layer (teal, upper-left panel of Fig.~\ref{fig:tesseract}) and its origins.

A native gate of width $\ell$ is an arbitrary unitary on its $\ell$ qubits, applied as one operation, so a faulty native gate may apply an arbitrary error to its whole support.
In Clifford layers each gate may be any Clifford unitary on its qubits.
Two parity rotations on a set $L$ of qubits appear throughout this work,
\begin{equation}
 T(L)=e^{-i\pi Z(L)/8},\qquad
 \sqrt X(L)=e^{-i\pi X(L)/4}.
 \label{eq:gates}
\end{equation}
Up to a global phase, $T(L)$ applies $T$ to the parity of the qubits in $L$, which is non-Clifford, and $\sqrt X(L)$ applies $\sqrt X$ to their parity in the $X$ basis, which is Clifford.
The non-Clifford layer applies either the gate $T(L)$ or $T(L)^\dagger$ of Eq.~\eqref{eq:gates} to each joining pair $L=\{x,x+e_4\}$, depending on whether $x_1+x_2+x_3$ is even or odd, respectively.
Up to a global phase, $T(L)$ multiplies a computational basis state $\ket b$ by $\omega^{b_x\oplus b_{x+e_4}}$, where $\omega=e^{i\pi/4}$ and $b_x\in\{0,1\}$ is the bit on vertex $x$.
In other words, it applies a $T$ gate to the parity of the pair.
On the code space, the product of the eight parity-gate phases reduces to $(-1)^{\lambda_{14}\lambda_{24}\lambda_{34}}$ for a codeword in which logical qubit $I$ holds the bit $\lambda_I$.
The layer therefore implements exactly $\overline{\CCZ}_{14,24,34}$, with no Clifford remainder.

It is known that the tesseract arises from the $\qcode{8}{3}{2}$ color code on the three-dimensional cube~\cite{reichardt2024}, the smallest interesting color code~\cite{kubica2015unfolding,campbell2016smallest}, which we call the cube.
The stabilizer group of the cube is generated by $X$ on all eight vertices and by $Z$ on each of the six faces. Its $X$ and $Z$ distances, the smallest weights of its $X$-type and $Z$-type logical operators, are $(d_X,d_Z)=(4,2)$.

The tesseract is obtained from the cube by placing one copy of the cube on the vertices $x_4=0$, placing a second copy, with $X$ and $Z$ exchanged, on the vertices $x_4=1$, and applying eight CNOTs, each controlled by a lower vertex $x+e_4$ and targeting its upper partner $x$.
The coupling combines the two codes, each with three logical qubits, into the six-logical-qubit tesseract and raises the distance from two to four.

In the computational basis, each CNOT adds the bit of the lower vertex to the bit of the upper one.
The bit $s_x$ of the cube's vertex $x$ therefore becomes the parity $b_x\oplus b_{x+e_4}$ of the joining pair.
The cube's non-Clifford gate, $T$ on the four even-weight vertices and $T^\dagger$ on the four odd-weight vertices, implements $\overline{\CCZ}$ on its three logical qubits~\cite{campbell2016smallest}.
The CNOT circuit thereby spreads this gate's phase, $\omega^{\sum_x(-1)^{x_1+x_2+x_3}s_x}$, to the phase $\omega^{\sum_x(-1)^{x_1+x_2+x_3}(b_x\oplus b_{x+e_4})}$ of the pair layer above, and the logical action is unchanged.
This explains why the tesseract inherits a two-fold transversal non-Clifford gate from the cube.

This observation is an instance of the general picture behind the $\ell$-fold transversal non-Clifford gates of all codes constructed in this work.
A seed code, here the cube, is coupled to a kernel code, here another cube, by a CNOT circuit.
The circuit spreads the seed's non-Clifford gate to an $\ell$-fold transversal gate, with $\ell=2$ in this case.
Section~\ref{sec:construction} makes this picture exact.

\section{Carbon code via seed-kernel construction}\label{sec:carbon}

As a second example of the seed--kernel procedure sketched in Fig.~\ref{fig:coupling}(b), we show that the $\qcode{12}{2}{4}$ Carbon code~\cite{knill2005quantum,paetznick2024demonstration}, listed as \texttt{eczoo:9a9707b2} in Table~\ref{tab:ps-codes}, also admits a universal set of gates with depth-one generators.
Eight CNOTs couple a distance-two seed to a four-qubit kernel state.
The distance doubles to four, and the seed's two-fold non-Clifford layer becomes three-fold transversal.
Together with four two-fold Clifford layers~\cite{albert2026beyond} on three other partitions, it gives a universal gate set.

The twelve qubits of the Carbon code are arranged as four triangles in Fig.~\ref{fig:coupling}(a).
Its checks are
\begin{align}
 X:\ &X_1X_3X_5X_7,\quad X_2X_4X_6X_8,\quad X_1X_2X_7X_8X_9X_{12},\nonumber\\
 &X_1X_2X_5X_6X_{10}X_{12},\quad X_1X_2X_3X_4X_{11}X_{12},\nonumber\\
 Z:\ &Z_1Z_6Z_7Z_8Z_{10}Z_{12},\quad Z_2Z_5Z_6Z_7Z_9Z_{12},\nonumber\\
 &Z_3Z_5Z_6Z_8Z_9Z_{11},\quad Z_4Z_5Z_7Z_8Z_{10}Z_{11},\nonumber\\
 &Z_9Z_{10}Z_{11}Z_{12}~.
 \label{eq:carbonchecks}
\end{align}
From now on we use the extended notation
\begin{equation}
 \code{n}{k}{d}{\underbrace{\Wt}_{\substack{\text{check}\\\text{weight}}}}{\underbrace{\ell}_{\substack{\text{gate}\\\text{width}}}}~,
 \label{eq:fiveparams}
\end{equation}
which also records the check weight $\Wt$, the largest weight of a check in the stabilizer generating set we give, and the native gate width $\ell$, the largest number of qubits on which one native gate of the generating layers acts.
In this notation, the Carbon code is a $\code{12}{2}{4}{6}{3}$ code.

Three of the four right panels of Fig.~\ref{fig:coupling}(a) draw and describe a minimal generating set of four two-fold transversal layers for the logical Clifford group
$\mathrm{Sp}(4,\F)$ modulo Paulis.
The fourth upper left panel, with teal lines, draws the code's non-Clifford layer,
\begin{equation}
 \prod_{\{a,b,c\}}S_a\,\CS^\dagger_{ab}\,\CS^\dagger_{ac}\,\CS_{bc}~,
 \label{eq:carbontri}
\end{equation}
where $\{a,b,c\}$ runs over the four triangles $\{1,2,12\}$, $\{3,4,11\}$, $\{5,6,10\}$, and $\{7,8,9\}$, with $c$ the kernel qubit, and $\CS_{pq}$ is the controlled-$S$ gate on qubits $p$ and $q$.
This is a depth-one layer of four disjoint diagonal three-qubit gates, one per triangle.
In the logical basis fixed below, this layer acts on the code space as $\bar S_1\bar S_2\,\overline{\CS}^\dagger_{12}$, a logical controlled-$S$ up to logical Clifford gates.

The seed code is the ``frozen cube'' $\qcode{8}{2}{2}$ on qubits $1,\ldots,8$.
It is the cube code $\qcode{8}{3}{2}$ of Sec.~\ref{sec:tesseract}, with its qubits relabeled and with one logical qubit fixed by adding the logical operator $X_1X_3X_5X_7$ to the stabilizer group.
Its stabilizers are
\begin{equation}
 \begin{array}{@{}l@{\qquad}l@{\qquad}l@{}}
  X_1X_3X_5X_7, & Z_1Z_6Z_7Z_8, & Z_2Z_5Z_6Z_7,\\
  X_2X_4X_6X_8, & Z_3Z_5Z_6Z_8, & Z_4Z_5Z_7Z_8,
 \end{array}
 \label{eq:frozencube}
\end{equation}
and its distances are $(d_X,d_Z)=(4,2)$.
Its qubits come in the four pairs $\{1,2\}$, $\{3,4\}$, $\{5,6\}$, $\{7,8\}$, which are the supports of its two-fold non-Clifford layer below and become the triangles of the Carbon code.
As logical operators we take $\bar X_1=X_1X_2X_5X_8$, $\bar X_2=X_1X_3X_4X_8$, $\bar Z_1=Z_2Z_3Z_5Z_8$, and $\bar Z_2=Z_1Z_2Z_3Z_4$, and we write $(\lambda_1,\lambda_2)$ for the bits held by the two logical qubits.

The kernel on qubits $9,\ldots,12$ is a $\qcode{4}{0}{2}$ stabilizer state in this case, so it encodes nothing.
It is the uniform superposition of the even-weight strings of length four, with stabilizers $X_9X_{12}$, $X_{10}X_{12}$, $X_{11}X_{12}$, and $Z_9Z_{10}Z_{11}Z_{12}$.

The depth-two coupling circuit $U_M$ of Fig.~\ref{fig:coupling}(b) consists of eight CNOTs, drawn in the left panel of Fig.~\ref{fig:coupling}(a).
Each kernel qubit forms a triangle with the pair it controls: $\{1,2,12\}$, $\{3,4,11\}$, $\{5,6,10\}$, and $\{7,8,9\}$.
Conjugating the stabilizers by the CNOTs leaves the seed $X$ checks and $Z_9Z_{10}Z_{11}Z_{12}$ unchanged.
It copies each kernel $X$ check onto the pairs its qubits control, and it adds to each seed $Z$ check a $Z$ on every kernel qubit that controls an odd number of its qubits.
The result is exactly Eq.~\eqref{eq:carbonchecks}, and the frozen cube's logical operators $\bar X_1$, $\bar X_2$, $\bar Z_1$, and $\bar Z_2$ remain logical operators of the coupled code, $\bar Z_1$ up to the kernel stabilizer $Z_9Z_{10}Z_{11}Z_{12}$.

The $X$ distance of the seed is retained.
The $Z$ distance doubles to four, because every weight-two $Z$ logical of the seed is copied onto the kernel qubits of the two triangles it touches, as in $Z_1Z_3\mapsto Z_1Z_3Z_{11}Z_{12}$.
The kernel does not supply any logicals since it is a state.

The frozen cube carries a depth-one pair layer that applies $S$ to the first member $a$ of each pair $\{a,b\}$ and $\CS^\dagger$ across the pair; one such pair, $\{7,8\}$, is shown in the circuit diagram of Fig.~\ref{fig:coupling}(a).
On a computational basis state with seed bits $s_1,\ldots,s_8$, the pair $\{a,b\}$ contributes the phase $i^{s_a-s_as_b}$, with $s_a$ from $S_a$ and $-s_as_b$ from $\CS^\dagger_{ab}$.
The basis strings of a codeword form a coset of the $X$ stabilizers, and the total phase is constant on every such coset, so the layer preserves the code.
It multiplies a codeword with logical bits $(\lambda_1,\lambda_2)$ by $i^{\lambda_1+\lambda_2-\lambda_1\lambda_2}$, which is $\bar S_1\bar S_2$ times a controlled-$\bar S^\dagger$.

The non-Clifford layer of Eq.~\eqref{eq:carbontri} is obtained by promoting a non-Clifford pair layer of the frozen cube to a layer of three-qubit gates.
The CNOTs rewrite the seed bits in terms of the physical bits of the Carbon code, i.e., the bits of a computational basis state after the coupling circuit.
In Sec.~\ref{sec:seedkernel}, each seed qubit $i$ together with the kernel qubits that control CNOTs onto it forms a set $F_i$, whose parity carries the seed bit $s_i$.
Here each kernel qubit controls CNOTs onto two seed qubits, so it lies in the sets $F_i$ of both, for example $F_1=\{1,12\}$ and $F_2=\{2,12\}$.
Each triangle is the union of two such sets and carries two seed bits.
On a computational basis state, the triangle $\{a,b,c\}$ with seed bits $s_a,s_b$ and kernel bit $u$ holds the physical bits $(x_a,x_b,x_c)=(s_a\oplus u,\ s_b\oplus u,\ u)$ after the CNOTs.
Its three edge parities are $x_a\oplus x_c=s_a$, $x_b\oplus x_c=s_b$, and $x_a\oplus x_b=s_a\oplus s_b$, and the kernel bit $u$ is absent in all of them.
A gate on the triangle whose phase depends only on these parities therefore acts on the seed bits alone, independently of the state of the kernel.

We can write the phase of the pair gate in terms of the physical bits of its triangle.
Substituting $s_a=x_a\oplus x_c$ and $s_b=x_b\oplus x_c$ gives
\begin{equation}
 s_a-s_as_b= x_a-x_ax_b-x_ax_c+x_bx_c
 \label{eq:carbonphase}
\end{equation}
as an identity between integers.
The right-hand side is the phase exponent of $S_a\,\CS^\dagger_{ab}\,\CS^\dagger_{ac}\,\CS_{bc}$: the $S$ gate contributes $i^{x_a}$, and each controlled-$S^{\pm1}$ gate contributes $i^{\pm x_px_q}$ on its pair $\{p,q\}$.
Applied to all four triangles, these gates form the depth-one layer of Eq.~\eqref{eq:carbontri}.
By Eq.~\eqref{eq:carbonphase}, its phase on every computational basis state equals the phase of the seed layer on the seed bits.
The layer therefore preserves the code and implements $\bar S_1\bar S_2\,\overline{\CS}^\dagger_{12}$, the same logical gate as the seed layer, now with width three.

\section{Seed-kernel construction of depth-1 universal codes}\label{sec:construction}
The tesseract and the Carbon code are two instances of one distance-increasing construction.
A small code, the seed, has a non-Clifford logical gate that is diagonal in the computational basis, e.g., the cube or frozen cube code for the two respective instances.
CNOTs from the qubits of a second code, the kernel, onto the seed qubits spread each seed qubit $i$ over a set $F_i$ of physical qubits, so that the parity of $F_i$ equals the seed bit $s_i$ on every encoded state.
A diagonal gate multiplies each basis state by a phase that depends only on its bits, so the seed gate keeps its logical action when its phase is evaluated on the parities of the sets $F_i$ instead of on the seed bits.
A single-qubit gate on seed qubit $i$ then acts on the parity of $F_i$, and a gate on several seed qubits acts on the parities of their sets.
A layer of disjoint diagonal gates on the seed thus becomes a layer of the same gates on the parities, of depth one when the sets of different gates share no qubit.

Gates diagonal in the $X$ basis, such as $\sqrt X$ rotations, transform in a different way under the circuit.
In the $X$ basis the roles of control and target flip, so the same CNOTs copy each seed qubit onto kernel qubits of its set $F_i$.
A gate diagonal in the $X$ basis can then act on any copy, as long as the kernel leaves the copies intact, which holds automatically when every kernel qubit starts in $\ket+$.
Placing different rotations on different copies can turn gates that overlap on the seed into disjoint gates.
For a kernel with $Z$ stabilizers, such as a product of GHZ states, only some choices of copies keep the logical action, and Sec.~\ref{sec:cliffordlayers} gives a condition on the kernel's $Z$ stabilizers that selects them.

Conversely, every CSS code with a diagonal logical layer that acts only on the parities of some sets of qubits, each set with a qubit that lies in no other set, arises in this way (Appendix~\ref{app:converse}).

Seed gates that are not diagonal in either basis, such as a transversal Hadamard, and any other missing Clifford layers, are found and checked for each code explicitly.
The kernel must be searched and chosen to make such gates available while maintaining or raising the distance.

We first define the seed code, the kernel code, and their coupling (Sec.~\ref{sec:seedkernel}), then derive the stabilizers and the distance of the coupled code (Sec.~\ref{sec:distance}), its non-Clifford layer (Sec.~\ref{sec:nonclifford}), and its Clifford layers (Sec.~\ref{sec:cliffordlayers}), and close with the catalog of codes (Sec.~\ref{sec:families}).

\subsection{Seed, kernel, and coupling}\label{sec:seedkernel}

The seed code $\mathcal S$ carries the non-Clifford gate of the construction, a logical gate that is diagonal in the computational basis, such as transversal $T$ on the 15-qubit Reed--Muller code. It is a CSS code on $m$ qubits, labeled $i=1,\ldots,m$. Its stabilizer group is generated by $X$-type operators $X(c)=\prod_{i\in c}X_i$ and $Z$-type operators $Z(c)=\prod_{i\in c}Z_i$ on sets $c$ of qubits, all with sign $+1$, which we call positive. We write $\mathcal S=\CSS(S_X,S_Z)$, where $S_X$ and $S_Z$ collect the supports of its $X$-type and $Z$-type stabilizers. For a collection $Y$ of sets, $Y^\perp$ denotes the sets that meet every member of $Y$ an even number of times, so $X(a)$ commutes with every $Z(c)$, $c\in Y$, exactly when $a\in Y^\perp$. The weight $\wt(P)$ of a Pauli operator $P$ is the number of qubits on which it acts nontrivially. The distances $d_X(\mathcal S)$ and $d_Z(\mathcal S)$ are the smallest weights of its $X$-type and $Z$-type logical operators, and $d(\mathcal S)$ is the smaller of the two.

The kernel code $\mathcal K=\CSS(K_X,K_Z)$ is a second CSS code, on $r$ further qubits labeled $q$, whose role is to raise the distance. A kernel without logical qubits is a stabilizer state. Common examples in our collection are the all-plus state $\ket+^{\otimes r}$, stabilized by every $X_q$ and by no $Z$-type operator, and products of GHZ states, which have $Z$ stabilizers as well.

The coupling of the seed and kernel consists of two CNOT circuits. In the first, $U_M$, kernel qubit $q$ controls a CNOT onto seed qubit $i$ when $M_{iq}=1$, for a binary matrix $M\in\F^{m\times r}$. Seed qubit $i$ and the kernel qubits that control CNOTs onto it form the set $F_i$, of size $|F_i|\le\ell$. Each seed qubit lies only in its own set. A kernel qubit that controls CNOTs onto several seed qubits lies in all of their sets, as in the Carbon code, and kernel qubits outside every $F_i$ take no part in $U_M$. The sets are disjoint exactly when each kernel qubit controls at most one CNOT, and we state where a result needs this. The second circuit, an optional twist $V_N$, acts before $U_M$ and applies a CNOT from seed qubit $i$ onto kernel qubit $q$ when $N_{qi}=1$, for $N\in\F^{r\times m}$. The coupled code is
\begin{equation}
 \mathcal Q=W(\mathcal S\otimes\mathcal K),\qquad W=U_MV_N.
 \label{eq:coupling}
\end{equation}
Its stabilizers and logical operators are those of $\mathcal S\otimes\mathcal K$, conjugated by $W$. A CNOT from qubit $c$ to qubit $t$ maps $X_c\mapsto X_cX_t$ and $Z_t\mapsto Z_cZ_t$, and fixes $X_t$ and $Z_c$. A twist $V_N$ therefore fixes every seed $Z_i$, and $U_M$ attaches to it the kernel qubits of its set, so
\begin{equation}
 WZ_iW^\dagger=Z(F_i)\qquad\text{for every seed qubit }i,
 \label{eq:zparity}
\end{equation}
for every twist and every kernel. In words, each seed operator $Z_i$ becomes the $Z$ parity of its set, so the parity of $F_i$ carries the seed bit $s_i$, as anticipated at the start of this section. A diagonal seed gate is a function of the operators $Z_i$, so the coupling carries it to the same function of the parities $Z(F_i)$ (Sec.~\ref{sec:nonclifford}).

The coupling also acts on the remaining operators, the $X$-type operators and the kernel operators $Z_q$. In the $X$ basis a CNOT acts with control and target exchanged, so $U_M$ copies the $X$-basis value of each seed qubit onto the kernel qubits of its set. In operator form, $U_M$ maps each kernel $X_q$ to $X_q$ times $X_i$ on the seed qubits of the sets that contain $q$. The twist maps each seed $X_i$ to $X_i$ times $X_q$ on the kernel qubits with $N_{qi}=1$, and each such kernel $Z_q$ to $Z_qZ_i$, which $U_M$ then turns into $Z_qZ(F_i)$. 
A twist therefore leaves Eq.~\eqref{eq:zparity} unchanged, so the seed's diagonal gate is carried over exactly as without it, but it has two other effects. First, it spreads the seed's $X$ logical operators onto kernel qubits, which changes the $X$-type logical operators of $\mathcal Q$ and can raise $d_X$. Second, when the kernel has $Z$ stabilizers, it multiplies their images by parities $Z(F_i)$, which restricts the copies that the $X$-basis gates discussed above may use. For the all-plus kernel neither effect occurs. Each twist CNOT then targets a qubit in $\ket+$, an eigenstate of $X$, so it acts as the identity.

\subsection{Stabilizers and distance}\label{sec:distance}
The stabilizer group of the coupled code $\mathcal Q$ is generated by the images $WGW^\dagger$ of the stabilizer generators $G$ of the seed and of the kernel. A CNOT circuit maps $X$-type operators to $X$-type operators and $Z$-type to $Z$-type, so $\mathcal Q$ is again a CSS code, and it encodes $k_{\mathcal S}+k_{\mathcal K}$ logical qubits, the sum of the numbers of logical qubits of the seed and of the kernel. Its logical operators are likewise the images $WPW^\dagger$ of the logical operators $P$ of $\mathcal S\otimes\mathcal K$. For a CSS code $\mathcal C$, let $\mathcal L_X(\mathcal C)$ be its set of nontrivial $X$-type logical operators, that is, the $X$-type operators that commute with every stabilizer but are not stabilizers. Define $\mathcal L_Z(\mathcal C)$ in the same way, and write $\mathcal G_Z(\mathcal C)$ for the group of its $Z$-type stabilizers. For the product code, $\mathcal L_X(\mathcal S\otimes\mathcal K)$ consists of the operators $X(a)\otimes X(u)$ in which $X(a)$ and $X(u)$ commute with the $Z$-type stabilizers of the seed and of the kernel, respectively. At least one of the two factors must not be a stabilizer. The distances of $\mathcal Q$ are therefore
\begin{flalign}
 &\eqindent \hspace{3.1em} d_X=\minover{P\,\in\,\mathcal L_X(\mathcal S\otimes\mathcal K)}\,\wt\bigl(WPW^\dagger\bigr),&&\label{eq:dxgeneral}\\[2pt]
 &\eqindent \hspace{3.1em} d_Z=\minover{P\,\in\,\mathcal L_Z(\mathcal S\otimes\mathcal K)}\,\wt\bigl(WPW^\dagger\bigr),&&\label{eq:dzgeneral}
\end{flalign}
and $d=\min(d_X,d_Z)$, because a logical operator with both $X$ and $Z$ parts has a nontrivial pure $X$ or pure $Z$ component on no larger support. Counting the parts of a partition that $WPW^\dagger$ meets, instead of its qubits, gives the partition distance of that partition (Sec.~\ref{sec:nonclifford}) in the same way. In binary form, where a set of qubits is identified with its indicator vector $v\in\F^{m+r}$, $W$ maps $X(v)$ to $X(Ev)$ and $Z(v)$ to $Z(E^{-T}v)$ with
\begin{equation}
 E=\begin{pmatrix}I+MN&M\\N&I\end{pmatrix}
 \label{eq:Ematrix}
\end{equation}
on the $m$ seed and $r$ kernel qubits. Appendix~\ref{app:binary} derives Eqs.~\eqref{eq:dxgeneral} and~\eqref{eq:dzgeneral}, together with their part-counting form, from this binary form.

Two consequences carry the construction. First, let $X(v)$ be a nontrivial $X$ logical of $\mathcal Q$ for a kernel without logical qubits. Undoing the coupling gives $W^\dagger X(v)W=X(a)\otimes X(u)$ with $X(u)$ a kernel stabilizer, so $X(a)$ is a nontrivial seed $X$ logical. By Eq.~\eqref{eq:zparity}, $X(v)$ anticommutes with $Z(F_i)$ exactly when $v$ meets $F_i$ an odd number of times. The set $a$ therefore consists of the seed qubits whose sets meet $v$ an odd number of times, which we call the seed support of $v$. For disjoint sets $\{F_i\}$, the seed support has at most as many qubits as $v$, so $d_X\ge d_X(\mathcal S)$. If every kernel qubit lies in at most $\kappa$ sets, the seed support can have up to $\kappa$ times as many qubits, and only $d_X\ge\lceil d_X(\mathcal S)/\kappa\rceil$ follows. Without a twist, $N=0$, the circuit $W$ fixes every seed $X_i$, so the seed's $X$ logicals keep their weight for every $M$ and $d_X\le d_X(\mathcal S)$. A twist also places these logicals on kernel qubits, which is how it can raise $d_X$.

Then, for an untwisted coupling with disjoint sets and a kernel without logical qubits, the two bounds meet, and
\begin{flalign}
 &\eqindent \hspace{3.1em}  d_X=d_X(\mathcal S),&&\nonumber\\
 &\eqindent \hspace{3.1em} d_Z=\minover{Z(z)\,\in\,\mathcal L_Z(\mathcal S)\\ Z(y)\,\in\,\mathcal G_Z(\mathcal K)}\,\wt\Bigl(Z(y)\prod_{i\,\in\,z}Z(F_i)\Bigr).&&\label{eq:untwisted}
\end{flalign}
The second line holds for every $M$, and the first needs disjoint sets. Appendix~\ref{app:statekernel} proves these bounds and Eq.~\eqref{eq:untwisted}. An untwisted coupling with disjoint sets thus keeps the $X$ distance of the seed and can only raise the $Z$ distance.

Second, Eq.~\eqref{eq:zparity} carries every nontrivial seed $Z$ logical $Z(z)=\prod_{i\in z}Z_i$ to
\begin{equation}
 W Z(z)W^\dagger=\prod_{i\,\in\,z}Z(F_i),
 \label{eq:zlogical}
\end{equation}
for any twist and kernel. This operator is a nontrivial logical of $\mathcal Q$ supported in the union of the sets $F_i$ with $i\in z$, where a kernel qubit that lies in an even number of them drops out. Therefore
\begin{flalign}
 &\eqindent \hspace{3.1em} d_Z\le\minover{Z(z)\,\in\,\mathcal L_Z(\mathcal S)}\,\sum_{i\,\in\,z}|F_i|\ \le\ \ell\,d_Z(\mathcal S).&&\label{eq:dbound}
\end{flalign}
For the 15-qubit seed with $\ell=3$, this gives $d\le9$, which the $\qcode{59}{1}{9}$ codes \texttt{lift:width3-59-w8} and \texttt{lift:width3-59} of Table~\ref{tab:ft-codes} attain. 

The all-plus kernel, $\ket+^{\otimes r}$, which has no $Z$ stabilizers, is the simplest choice and gives the largest $d_Z$ for given $M$ and $N$. A kernel with $Z$ stabilizers can only lower $d_Z$ (Appendix~\ref{app:statekernel}) and constrains the $X$-diagonal layers (Sec.~\ref{sec:cliffordlayers}).
However, a kernel with $Z$ stabilizers can lower the check weight, as in the $\qcode{37}{1}{7}$ code \texttt{lift:ghz37} of Sec.~\ref{sec:example}, where four GHZ states in place of an all-plus kernel lower the optimal $\Wt$ from ten to eight while preserving $d=7$ and the gates (Appendix~\ref{app:37dist}).

\subsection{The non-Clifford layer and partition distance}\label{sec:nonclifford}
Let $D$ be a logical gate of the seed that is diagonal in the computational basis. It is a function $D=f(Z_1,\ldots,Z_m)$ of the commuting operators $Z_i$, and we call
\begin{equation}
 D_\oplus=f\bigl(Z(F_1),\ldots,Z(F_m)\bigr)
 \label{eq:paritylayer}
\end{equation}
its parity layer. Conjugating $D\otimes I$ by $W$ replaces each $Z_i$ by $Z(F_i)$, by Eq.~\eqref{eq:zparity}, so
\begin{equation}
 D_\oplus\,W=W\,(D\otimes I)
 \label{eq:paritytransfer}
\end{equation}
for every twist and every kernel [Fig.~\ref{fig:coupling}(a) and (b), and Appendix~\ref{app:ztransfer} for the proof on the full Hilbert space]. The parity layer $D_\oplus$ therefore preserves $\mathcal Q$, has the logical action of the seed gate $D$, and acts as the identity on any logical qubits of the kernel. For transversal $T$ on the seed, $D_\oplus=\prod_iT(F_i)$ up to a global phase, with $T(L)$ the parity rotation of Eq.~\eqref{eq:gates}. More generally, for a seed gate $D=\bigotimes_iD_i$ of single-qubit gates and disjoint sets $\{F_i\}$, $D_\oplus$ is the disjoint layer that applies $D_i$ to the parity of $F_i$, of width $\max_i|F_i|$. If instead $D=\prod_jD_j$, where each $D_j$ is diagonal and supported on a set $G_j$ of seed qubits, with the $G_j$ disjoint, then $D_\oplus=\prod_j(D_j)_\oplus$, and $(D_j)_\oplus$ applies $D_j$ to the parities of the sets $F_i$ with $i\in G_j$, as one gate on their union. These gates are disjoint when sets $F_i$ that share a qubit have their indices in the same $G_j$, and the width is the size of the largest union. In the Carbon code of Sec.~\ref{sec:carbon}, each pair gate becomes a gate on a triangle, the union of the two sets of the pair.

The \textit{partition distance} $\Dpart(\mathcal Q;\mathcal B)$ of a layer is the fewest parts of $\mathcal B$ that together carry a nontrivial logical operator, where the partition $\mathcal B$ consists of the gate supports of the layer and a singleton for every idle qubit. 
It is the min-distance of Ref.~\cite{jochymoconnor2018disjointness} with respect to $\mathcal B$ and quantifies fault tolerance: arbitrary faults on any $s$ gates are always correctable exactly when $\Dpart\ge2s+1$ (Appendix~\ref{app:nativefaults}).
A faulty native gate can corrupt all of its qubits at once, which is why the partition distance rather than the code distance governs a layer. For gates on at most $\ell$ qubits, $\lceil d/\ell\rceil\le\Dpart\le d$ (Appendix~\ref{app:nativefaults}). A code with $d\ge2\ell+1$ therefore has $\Dpart\ge3$ in every generating layer and corrects one arbitrary faulty native gate in each.

For example, the triangle layer of the Carbon code (Sec.~\ref{sec:carbon}) has $\Dpart=2$. The logical operator $Z_1Z_3Z_{11}Z_{12}$ is supported on the two triangles $\{1,2,12\}$ and $\{3,4,11\}$. Two is also a lower bound, since a logical operator of weight $d=4$ meets at least $\lceil d/3\rceil=2$ gates of width three. 

The partition distance of $D_\oplus$, on the partition $\mathcal B$ of its gate supports, can be computed exactly by Eqs.~\eqref{eq:dxgeneral} and~\eqref{eq:dzgeneral} with the weight replaced by the number of parts of $\mathcal B$ that the support of $WPW^\dagger$ meets. By Eq.~\eqref{eq:zlogical}, the image $WZ(z)W^\dagger$ of a seed $Z$ logical lies in the union of the sets $F_i$ with $i\in z$. Hence
\begin{flalign}
 &\eqindent \hspace{5.9em} \Dpart(\mathcal Q;\mathcal B)\le d_Z(\mathcal S)&&\label{eq:ceiling}
\end{flalign}
for every partition $\mathcal B$ in which each set $F_i$ lies inside one part, whatever the twist and the kernel. The gate supports of $D_\oplus$ form such a partition when the seed gate $D$ acts on every seed qubit. For disjoint sets, $\mathcal B$ can also be the sets $F_i$ together with singletons for the other kernel qubits. 
When the coupling is untwisted, the sets are disjoint, and the kernel has no logical qubits, the partition into the sets $F_i$, with singletons for the other kernel qubits, has exactly the distance of the seed,
\begin{flalign}
 &\eqindent  \hspace{5.9em} \Dpart(\mathcal Q;\{F_i\})=d(\mathcal S)&&\label{eq:untwistedblock}
\end{flalign}
(Appendix~\ref{app:nativefaults}). A seed with $d_Z(\mathcal S)=3$ therefore gives a parity layer $D_\oplus$ that cannot correct two faulty gates, however large the kernel makes the code distance. The $\qcode{37}{1}{7}$ code \texttt{lift:ghz37} of Sec.~\ref{sec:example} raises $d$ from three to seven and retains $\Dpart=3$ on its $D_\oplus$.

\subsection{Clifford layers}\label{sec:cliffordlayers}
While the seed-kernel construction ensures that a non-Clifford depth-one layer is preserved on the output code $\mathcal Q$, it does not by itself guarantee a set of Clifford generators that completes to a universal gate set. 

Diagonal Clifford gates of the seed, such as $S$ and CZ, transfer as parity layers by Eq.~\eqref{eq:paritytransfer}, for every kernel and every twist.
Gates diagonal in the $X$ basis are guaranteed to transfer for particular kernels.
For a set $L$ of seed qubits, the rotation $\sqrt X(L)=e^{-i\pi X(L)/4}$ of Eq.~\eqref{eq:gates} is such a gate, and a product $\prod_j\sqrt X(L_j)$ of rotations on sets $L_j$ of seed qubits can implement a logical Clifford gate even when the sets overlap.
When the kernel is all-plus, every logical seed gate that is diagonal in the $X$ basis transfers for every twist, through copies instead of parities (Lemma~\ref{lem:xtransfer} in Appendix~\ref{app:xtransfer}).
For a product of rotations, replace each seed qubit $i\in L_j$ by any qubit of $F_i$ that lies in no other set, chosen independently for each $j$.
The layer $\prod_j\sqrt X(L'_j)$ on the resulting sets $L'_j$ then implements the same logical gate on $\mathcal Q$.
The sets $L'_j$ can be made disjoint, so that the layer has depth one, exactly when no seed qubit $i$ lies in more of the sets $L_j$ than $F_i$ has qubits in no other set.

For kernels with $Z$ stabilizers but no logical qubits, consider a layer $\prod_j\sqrt X(L'_j)$ that implements a logical gate when the kernel is replaced by $\ket+^{\otimes r}$, with the same $M$ and $N$.
The layer preserves $\mathcal Q$ and keeps its logical action if every kernel $Z$ stabilizer $Z(y)\in\mathcal G_Z(\mathcal K)$, carried through the coupling, commutes with every rotation generator,
\begin{equation}
 \bigl[\,WZ(y)W^\dagger,\,X(L'_j)\,\bigr]=0\quad\text{for all }j,
 \label{eq:constraint}
\end{equation}
(Appendix~\ref{app:constraint}). A $Z$-type operator commutes with $X(L'_j)$ exactly when its support meets $L'_j$ an even number of times, so Eq.~\eqref{eq:constraint} says that no kernel $Z$ stabilizer, carried through the coupling, changes the parity of any gate support. The condition is sufficient but not necessary.
Conjugation by the layer multiplies a $Z$ stabilizer by $\pm iX(L'_j)$ for each rotation generator that anticommutes with it.
The layer therefore still preserves $\mathcal Q$, with the same logical action, when these factors multiply to an $X$ stabilizer of $\mathcal Q$ (Appendix~\ref{app:constraint}).
Without a twist, $W$ fixes every kernel $Z$ stabilizer, and Eq.~\eqref{eq:constraint} says that every kernel $Z$ stabilizer has even overlap with every gate support.

\subsection{A catalog of depth-one universal codes}\label{sec:families}

The general construction discussed so far yields a search procedure consisting of four steps. 
First, choose a seed whose $d_Z(\mathcal S)$ gives the required ceiling of Eq.~\eqref{eq:ceiling} on the parity layer. 
Second, choose the sets $F_i$ and the supports of the missing Clifford layers. 
Third, choose the twist and the kernel $Z$ stabilizers so that Eq.~\eqref{eq:constraint} holds, evaluate both distance minima, and find a generating set of low check weight. 
Fourth, verify the logical group~\cite{albert2026beyond} and the partition distance of each native partition.

Varying the seed, the kernel, and the sets $F_i$, we find \NDB{} depth-one universal codes, with distances from three to 49.
Table~\ref{tab:ft-codes} lists and Figure~\ref{fig:concept}(b) plots the \NFT{} codes that meet $d\ge2\ell+1$ and so correct one faulty gate per layer (Sec.~\ref{sec:nonclifford}), and Table~\ref{tab:ps-codes} lists the \NPS{} codes that only detect one.
The accompanying JSON file gives every generator, Clifford and non-Clifford, in the extended qiskit format of Ref.~\cite{albert2026beyond}, which builds on the code database of Qiskit QEC~\cite{crossqiskitqec2025,cross2025small}.
That file also lists the seed code, kernel code, and CNOT circuits for each depth-one universal code.
\input{tab_ftcodes}

Of the \NDB{} codes, 42 grow from the cube $\qcode{8}{3}{2}$~\cite{kubica2015unfolding,campbell2016smallest} and six from the frozen cube $\qcode{8}{2}{2}$ of Sec.~\ref{sec:carbon}.
Another 25 grow from the 15-qubit quantum Reed--Muller code $\qcode{15}{1}{3}$~\cite{knill1996,bravyi2005}, nine from the 64-qubit quantum Reed--Muller code $\qcode{64}{15}{4}$~\cite{rengaswamy2020optimality,hu2025climbing}, and three from the 31-qubit quantum Reed--Muller code $\qcode{31}{1}{3}$~\cite{anderson2014,landahl2013complex}.
Eight grow from a $\qcode{28}{2}{3}$ seed, obtained by gluing two copies of the 15-qubit code, one carrying $T$ and the other $T^\dagger$, along one pair of coordinates, which are then removed.
The remaining 17 codes come from ten further seeds, from the hypercube code $\qcode{16}{4}{2}$~\cite{vasmer2022morphing,hangleiter2025fault,barg2026geometric} to a $\qcode{512}{84}{8}$ code.

Native widths range from two to ten.
Of the \NDB{} codes, 35 have width two, 53 width three, seven width four, two width five, ten width six, one width nine, and two width ten.
An entry of particular width does not by itself rule out a generating set at a smaller width, and indeed we were able to find such narrower generating sets for several codes (discussed below).

Check weights run from six to 128, and 42 codes have checks of weight at most eight.
At width two the largest distance is eight, at widths three and four sixteen, at width five fifteen, and at width six 30.
Two codes stand out for high distance at low width. The code $\code{37}{1}{7}{8}{3}$\,\texttt{lift:ghz37} has the lowest check weight of the three 37-qubit codes, the smallest codes that correct a faulty gate in every layer. The width-nine code $\code{1369}{1}{49}{56}{9}$\,\texttt{lift:ghz37sq} has distance 49, more than five times its width.
Two codes, marked by a dagger in Table~\ref{tab:ft-codes}, are unions of two codes and are universal only through layers that couple the two.

By Lemma~\ref{lem:goursat} in Appendix~\ref{app:converse}, every code is the coupling of the seed listed with it to a kernel, so Eqs.~\eqref{eq:dxgeneral} and~\eqref{eq:dzgeneral} apply to all of them.
The bound of Eq.~\eqref{eq:dbound}, with $\ell$ the size of the largest set $F_i$ of that coupling, holds for all \NDB{} codes and is attained by 43 of them.
The partition-distance results, Eqs.~\eqref{eq:ceiling} and~\eqref{eq:untwistedblock}, concern the parity layer of that coupling.
For all but eight codes the parity layer is the recorded non-Clifford layer, so both results apply to it directly.

The eight exceptions are \texttt{uqlc:\allowbreak PCR121} and \texttt{uqlc:\allowbreak PCR128}, whose seed is the cube, and \texttt{uqlc:\allowbreak GFP383}, \texttt{uqlc:\allowbreak GFP516}, \texttt{uqlc:\allowbreak GFP705}, \texttt{uqlc:\allowbreak GFP754}, \texttt{uqlc:\allowbreak GFP761}, and \texttt{uqlc:\allowbreak GFP768}, whose seed is $\qcode{64}{15}{4}$.
For each of them the search found a narrower layer, made of three-qubit gates, that acts as the seed gate on the code space but is not a function of the set parities.
This narrower layer is the one recorded in the JSON file and in Table~\ref{tab:ft-codes}.
The parity layers prescribed by our construction exist for these codes as well.
Their sets have four qubits, five in one set of \texttt{uqlc:PCR121}, and they obey Eqs.~\eqref{eq:dbound} and~\eqref{eq:ceiling}, with \texttt{uqlc:PCR128} and \texttt{uqlc:GFP768} reaching $d=\ell\,d_Z(\mathcal S)$.

The tesseract of Sec.~\ref{sec:tesseract} shows the construction with a kernel that encodes logical qubits.
Its seed is the cube, with $S_X$ spanned by the set $\mathbf1_8$ of all eight vertices and $S_Z=\mathrm{RM}(1,3)$, the first-order Reed--Muller code spanned by $\mathbf1_8$ and the six faces, and its kernel is the same code with $X$ and $Z$ exchanged, coupled coordinate by coordinate with $M=I$ and $N=0$.
Coupling a seed to a copy of itself raises the distance and adds logical qubits but keeps the partition distance of the seed's gate.
Number the seed vertices $1,\ldots,8$ and their kernel partners $9,\ldots,16$ from left to right as in Fig.~\ref{fig:coupling}(c). Adding the twist $\CNOT(1\to9)\,\CNOT(1\to10)$ before the coupling changes only the $X$ checks and gives a second inequivalent $\code{16}{6}{4}{8}{2}$ code, \texttt{copycup:66fe2995} in Table~\ref{tab:ps-codes}, shown in Fig.~\ref{fig:coupling}(c).

\begin{figure}[t]
\centering\includegraphics[width=\columnwidth]{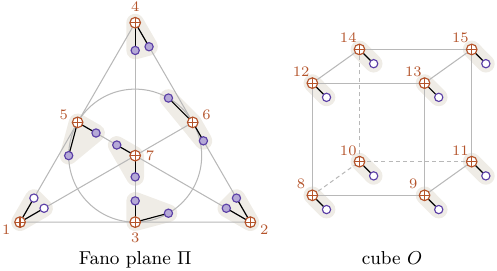}
\caption{The qubits of \texttt{lift:ghz37}. Seed qubits (orange) are the points of the Fano plane $\Pi$ and the cube $O$, each the target ($\oplus$) of CNOTs (shaded) from the kernel qubits (purple) of its set $F_i$. The four GHZ triples (filled) lie on the Fano lines not through point~$1$, and open dots are kernel qubits in $\ket+$.}
\label{fig:lift37}
\end{figure}

\section{A small fault-tolerant depth-one universal code}\label{sec:example}

The tesseract and the Carbon code detect a faulty native gate but cannot correct one. A code corrects one in every generating layer when $d\ge2\ell+1$ (Sec.~\ref{sec:nonclifford}), and the \NFT{} codes of Table~\ref{tab:ft-codes} meet this condition. The smallest of them have 37 qubits, among them \texttt{lift:ghz37} with parameters $\Qh=\code{37}{1}{7}{8}{3}$. We now describe $\Qh$ and provide a distance-preserving syndrome extraction circuit.
Despite its small distance of seven, the code performs well under circuit-level noise.

The code $\Qh$ couples the 15-qubit Reed--Muller code~\cite{knill1996,bravyi2005}, with $(d_X,d_Z)=(7,3)$ and transversal $T$ acting as $\bar T^\dagger$, to a 22-qubit kernel of four GHZ triples and ten qubits in $\ket+$ [Fig.~\ref{fig:coupling}(c)]. Untwisted coupling CNOTs join each seed qubit $i$ to one or two kernel qubits in a set $F_i$, which gives seven sets of three qubits and eight of two (Fig.~\ref{fig:lift37}). The GHZ kernel lowers the optimal check weight from ten, its value for the all-plus kernel, to eight while keeping $d=7$, and a final CNOT circuit leaves at most six checks on every qubit (Appendix~\ref{app:37seed}). 
The accompanying JSON file of Sec.~\ref{sec:families} lists the code and its generators, and Appendix~\ref{app:37} proves $d=7$ and the optimality of $\Wt=8$.

The generators of $\Qh$ are four depth-one layers. The parity layer $\prod_iT(F_i)$, of seven three-qubit and eight two-qubit gates, implements $\bar T^\dagger$ by Eq.~\eqref{eq:paritytransfer}. 
A layer of $\sqrt X(L)$ on seven disjoint triples $L$ implements $\sqrt{\bar X}^{\,\dagger}$ through copies, as in Sec.~\ref{sec:cliffordlayers}. The seed supports of these triples are the seven lines of a Fano plane formed by seven seed qubits. Bitwise CNOT between two blocks implements $\overline{\CNOT}$. Since $\bar H\propto\bar S^\dagger\sqrt{\bar X}^{\,\dagger}\bar S^\dagger$, these layers generate the logical Clifford group, and $\bar T^\dagger$ completes a universal set. Both single-block layers have $\Dpart=3$ and bitwise CNOT has $\Dpart=7$, so every generating layer corrects one arbitrary faulty native gate (Appendix~\ref{app:37gates}). For the parity layer, three is the ceiling $d_Z(\mathcal S)$ of Eq.~\eqref{eq:ceiling}, so two faulty gates in that layer can already cause a logical error.

A distance-seven code protects a computation only if its syndrome extraction does not turn a few faults into a logical error. The checks of $\Qh$ reach weight eight, so a single ancilla fault could spread to a high-weight data error, and we add flag qubits~\cite{chamberland2018} to 14 of the 36 checks. One interleaved round measures all checks, with one ancilla per check and 15 flag qubits, each used by a single check, on 88 qubits with 192 CNOTs in 12 layers. Among such rounds, it is optimal in flags, CNOTs, layers, and qubits at once. The round preserves the distance: its circuit distance, the fewest faults that flip the logical without flipping any detector, is $d_{\rm circ}=d=7$ for any number of rounds. Appendix~\ref{app:extraction} gives the circuit and the proofs of $d_{\rm circ}=7$ and of its optimality.

We simulate memory and every generator with Stim~\cite{gidney2021stim} and decode with the Tesseract decoder~\cite{tesseract2025}. All gate, reset, and measurement errors are tied to the two-qubit gate error $p$. A three-qubit native gate fails as often as a CNOT, which is optimistic, since demonstrated three-atom gates are less accurate than two-qubit gates~\cite{levine2019,evered2023}. Idle qubits decay with the coherence times $T_1=30$\,s and $T_2=10$\,s assumed in Ref.~\cite{zhou2025}. Each experiment prepares every data qubit in $\ket0$ or in $\ket+$, runs cycles of one gate layer and one extraction round, and reads out in the same basis (Appendix~\ref{app:noise}). 

\begin{table}[t]
\caption{Logical failure of $\Qh$ at $p=10^{-3}$ in the noise model of Appendix~\ref{app:noise}, with coherence times $T_1=30$\,s and $T_2=10$\,s, per round for memory and per cycle of one layer and one round for the generators. The $X$ and $Z$ columns prepare and read out in that basis. The last column adds these two marginal rates, so a logical $Y$ error is counted in both. The $X$-basis entry of $\bar T^\dagger$ ends with a noiseless round and readout, and its $Z$-basis entry is a quarter of the failure of four cycles. The $Z$-basis slope of $\bar S^\dagger$, $(-0.3\pm0.6)\times10^{-7}$, is consistent with zero.}
\label{tab:37rates}\centering\small\setlength{\tabcolsep}{4pt}
\begin{tabular}{lccc}\toprule
 & $X$ basis & $Z$ basis & Sum\\\midrule
Memory & $1.5\times10^{-7}$ & $1.6\times10^{-7}$ & $3.1\times10^{-7}$\\
$\bar S^\dagger$ & $2.6\times10^{-5}$ & $<10^{-7}$ & $2.6\times10^{-5}$\\
$\sqrt{\bar X}^{\,\dagger}$ & $2.4\times10^{-7}$ & $2.8\times10^{-6}$ & $3.1\times10^{-6}$\\
$\overline{\CNOT}$ & $6.6\times10^{-7}$ & $3.4\times10^{-7}$ & $1.0\times10^{-6}$\\
$\bar T^\dagger$ & $2.5\times10^{-4}$ & $6.7\times10^{-7}$ & $2.5\times10^{-4}$\\\bottomrule
\end{tabular}\end{table}

The $\bar T^\dagger$ layer is not Clifford, but it maps a Pauli error to the same error times a Clifford correction on the sets $F_i$ that its $X$ part meets an odd number of times. We therefore simulate it exactly with stabilizer methods~\cite{takada2026}. A classical controller decodes the record before each layer and applies the correction on the sets that it estimates odd. It misjudges at least one set in about one layer in ten. The final decoder also sees the later rounds, and on every set whose decision it would revise it allows a $Z$ error at no cost, which lowers the failure about tenfold. 

Every generator of $\Qh$ thus runs as one depth-one layer followed by one extraction round, with the failures of Table~\ref{tab:37rates}. The single-block layers limit this performance. Their partition distance is only $\Dpart=3$, so two faulty native gates can already cause a logical error, and the cycles of $\bar S^\dagger$ and $\bar T^\dagger$ fail about 80 and 800 times as often as a memory round. Preliminary simulations of the $\code{59}{1}{9}{8}{3}$ code \texttt{lift:width3-59-w8} in the same noise model support this picture. Its memory and $\overline{\CNOT}$ fail at most about $2\times10^{-8}$ per round and $1.5\times10^{-7}$ per cycle, seven to fourteen times less often than those of $\Qh$. The single-block layers, whose partition distance is still $\Dpart=3$, do not share this gain. The $\bar S^\dagger$ and $\bar T^\dagger$ layers fail about $1.4\times10^{-5}$ and $1.2\times10^{-4}$ per cycle, within a factor of three of those of $\Qh$. The $\sqrt{\bar X}^{\,\dagger}$ layer fails at most about $6\times10^{-5}$ per cycle, more often than in $\Qh$. Tesseract is too slow for this code, so its memory and Clifford layers are decoded by belief propagation, with Tesseract on the 3 to 7\% of shots where it does not converge. 

For seed codes with diagonal non-Clifford gates (``the parity layer''), a larger partition distance needs a seed with a larger $Z$ distance by Eq.~\eqref{eq:ceiling}. Table~\ref{tab:ft-codes} contains six codes with $d\ge4\ell+1$, for which $\Dpart\ge\lceil d/\ell\rceil\ge5$ in every generating layer (Appendix~\ref{app:nativefaults}), so that each layer corrects two arbitrary faulty native gates. They are \texttt{uqlc:\allowbreak GFP705}, \texttt{uqlc:\allowbreak GFP754}, \texttt{uqlc:\allowbreak GFP761}, and \texttt{uqlc:\allowbreak GFP768} at width three, \texttt{lift:\allowbreak proj570} at width six, and \texttt{lift:\allowbreak ghz37sq} at width nine. Their checks have weights from 36 to 96, so their extraction circuits are larger, and we plan to simulate them under circuit-level noise in future work.

\section{Conclusions}\label{sec:conclusion}
The Eastin--Knill theorem forbids a universal set of transversal logical gates on any code that detects single-qubit errors~\cite{eastin2009}. It fixes one partition of the qubits in advance and constrains only the gates that respect it. Depth-one layers of few-qubit gates on \textit{different} partitions of the same code escape this constraint. The layers of any one partition generate a finite logical group when every part has fewer than $d$ qubits, but layers on several partitions can generate a universal one. We call such codes depth-one universal: every logical generator is one layer followed by error correction on a fixed code, with no magic-state distillation, code switching, or lattice surgery. The 16-qubit tesseract code realizes this with four layers of two-qubit gates.

We develop a procedure that raises the distance of a code with a depth-one non-Clifford gate while increasing the width of that gate, the number of qubits it acts on, as little as possible. A CNOT circuit couples the code (the seed) to an auxiliary code (the kernel) so that each seed qubit becomes the parity of a small group of qubits. A diagonal seed gate then acts on these parities with the same logical action, and its width grows only to the size of the largest group. The kernel raises the distance, by Eq.~\eqref{eq:dbound} up to the width times the $Z$ distance of the seed (the weight of its lightest $Z$-type logical operator). 

We construct \NDB{} depth-one universal codes with 12 to 1369 qubits, one to 210 logical qubits, distances from three to 49, and widths from two to ten. Of these, \NFT{} codes have every generating layer, Clifford and non-Clifford alike, of width $\ell\le\lfloor(d-1)/2\rfloor$, so one faulty gate in any layer leaves a correctable error. They include the $\qcode{37}{1}{7}$ code \texttt{lift:ghz37} and the $\qcode{59}{1}{9}$ code \texttt{lift:width3-59-w8} of Table~\ref{tab:ft-codes}, both with checks of weight at most eight, and codes with 70 and 210 logical qubits. 

For this $\qcode{37}{1}{7}$ code, with its check weights redistributed by a further CNOT circuit, we give a syndrome-extraction circuit with flag qubits that preserves its distance, and we evaluate memory and every generator under circuit-level noise. Memory and the $\overline{\CNOT}$ and $\sqrt{\bar X}^{\,\dagger}$ layers fail rarely, but the $\bar S^\dagger$ and $\bar T^\dagger$ layers fail much more often. This reflects the low partition distance $\Dpart=3$ of the parity layer: two faulty gates, or for $\bar T^\dagger$ two misjudged corrections, can already cause a logical error.

Our codes grow from sixteen seeds and have check weights from six to 128. We anticipate that many more depth-one universal codes can be constructed, and Sec.~\ref{sec:families} gives a search recipe. We are particularly excited about the prospect of high-rate quantum LDPC codes with such gates. Our codes with 35 or more logical qubits all have checks of weight 32 or more. A family of bounded width, whose gates spread one fault to boundedly many qubits, must evade the no-go theorem of Ref.~\cite{webster2022universal}. That theorem excludes concatenated and conventional topological codes but leaves room for codes with distance proportional to length~\cite{jochymoconnor2018disjointness}. Quantum LDPC codes with transversal or constant-depth non-Clifford gates~\cite{lin2024transversal,zhu2025nonclifford,golowich2025quantum,breuckmann2026cups,menon2026magic,golowich2026cup} may serve as seeds. Codes of higher distance should also yield fully unitary fault-tolerant compilers, with the non-Clifford generator applied directly to the data, at the error rates near $10^{-3}$ of current two-qubit gates~\cite{evered2023,evered2026}. This requires a large partition distance in every generating layer, not only a large code distance. Finding codes of practical size with this property and with low-weight stabilizers is therefore an important direction. Every layer corrects two faulty gates once $d\ge4\ell+1$, which six codes of Table~\ref{tab:ft-codes} satisfy, but a gate on parities also needs a seed with $Z$ distance at least five [Eq.~\eqref{eq:ceiling}].

The point of depth-one universal codes is to avoid magic-state factories altogether. If the physical error rate is too high for the direct non-Clifford layer at the distances of our codes, however, factories built from the same codes remain promising. A separate block applies the non-Clifford layer and is kept only if no error is detected. A transversal CNOT~\cite{huang2023homomorphic} then teleports the resulting magic state into a data block, in the role that magic-state cultivation~\cite{gidney2024cultivation} plays for the surface code. For the $\qcode{37}{1}{7}$ code at $p=10^{-3}$, half of the attempts are accepted, and an accepted state fails with an estimated $5.8\times10^{-12}$ (Appendix~\ref{app:magic}). The teleported gate is then limited by the transversal CNOT and the Clifford correction, whose rates in Table~\ref{tab:37rates} lie well below that of the direct layer.

\begin{acknowledgments}
VVA acknowledges NSF grants OMA-2120757 (QLCI) and CIF-2330909. Parts of this manuscript were prepared with the assistance of
ChatGPT 5.6, developed by OpenAI, and Claude Code, developed by Anthropic. These agents were used to run numerical searches, verify
proofs, and edit the writing of some sections, in accordance with the authors' instructions. All content, claims, and conclusions have
been reviewed and verified by the authors to ensure accuracy and originality. Certain products, commercial and otherwise, are mentioned
in this publication. These mentions are for informational purposes only, and do not imply recommendation or endorsement by NIST.
\end{acknowledgments}

\bibliography{references}

\clearpage
\appendix

\section{Proofs for the seed--kernel construction}\label{app:theorems}
\input{tab_pscodes}
This appendix proves the results of Sec.~\ref{sec:construction} in the order in which they appear there. Appendix~\ref{app:binary} writes the coupling in binary form and derives the distance minima, and Appendix~\ref{app:statekernel} evaluates them for kernels without logical qubits. Appendices~\ref{app:ztransfer} and~\ref{app:nativefaults} prove the transfer of diagonal seed gates and the correctability of faulty parts. Appendices~\ref{app:xtransfer} and~\ref{app:constraint} treat gates diagonal in the $X$ basis, and Appendix~\ref{app:converse} proves the converse.

\subsection{Binary form of the coupling}\label{app:binary}
We write the coupling in binary form, which gives the check spaces of $\mathcal Q$ and proves Eqs.~\eqref{eq:dxgeneral} and~\eqref{eq:dzgeneral}. The product code's $X$ normalizer and stabilizer spaces are $S_Z^\perp\oplus K_Z^\perp$ and $S_X\oplus K_X$. The $Z$ spaces are $S_X^\perp\oplus K_X^\perp$ and $S_Z\oplus K_Z$. The CNOT circuit $W$ permutes computational basis states, $\ket v\mapsto\ket{Ev}$, and sends the supports of pure $X$-type and $Z$-type operators to $Ev$ and $E^{-T}v$. It preserves ranks, commutation, and logical equivalence, which gives the check spaces below and the two minima, as well as their part-counting form.

For an $X$-type operator $X(a)\otimes X(u)$ and a $Z$-type operator $Z(z)\otimes Z(y)$ of the product code, with seed parts $a,z\in\F^m$ and kernel parts $u,y\in\F^r$, let $\mathcal X_N(a,u)=E(a,u)$ and $\mathcal Z_N(z,y)=E^{-T}(z,y)$ be the supports of their images under $W$. For $E$ of Eq.~\eqref{eq:Ematrix}, these maps and the inverse transpose are
\begin{align}
 E^{-T}&=\begin{pmatrix}I&N^T\\M^T&I+M^TN^T\end{pmatrix},\nonumber\\
 \mathcal X_N(a,u)&=(a+M(Na+u),Na+u),\nonumber\\
 \mathcal Z_N(z,y)&=(z+N^Ty,M^T(z+N^Ty)+y).
 \label{eq:imagesapp}
\end{align}
The full check spaces are $Q_X=\mathcal X_N(S_X,K_X)$ and $Q_Z=\mathcal Z_N(S_Z,K_Z)$, with no restriction on $N$. At $N=0$ they become
\begin{equation}
 Q_X=\{(g+Mb,b)\},\qquad Q_Z=\{(z,M^Tz+y)\},
\end{equation}
with $g\in S_X$, $b\in K_X$, $z\in S_Z$, and $y\in K_Z$.

\subsection{Kernels without logical qubits}\label{app:statekernel}
For $K_X=K_Z^\perp$, a product logical is nontrivial exactly when its seed factor is nontrivial. Substituting $v=Na+u$ in Eqs.~\eqref{eq:dxgeneral} and~\eqref{eq:dzgeneral}, the minima become
\begin{align}
 d_X&=\min_{\substack{a\in S_Z^\perp\setminus S_X\\v\in Na+K_Z^\perp}}
       [\wt(a+Mv)+\wt(v)],\label{eq:dxstate}\\
 d_Z&=\min_{\substack{z\in S_X^\perp\setminus S_Z\\y\in K_Z}}
       \wt\!\left(E^{-T}\binom zy\right).\label{eq:dzstate}
\end{align}
Here $E^{-T}(z,0)$ is $Z$ on the seed qubits of $z$ and on every kernel qubit that lies in an odd number of the sets $F_i$ with $i\in z$. For disjoint sets it is the seed operator $Z(z)$ copied onto every qubit of these sets. For every $M$, $E^{-T}(0,y)$ is the physical image of the kernel stabilizer $Z(y)$. Taking the seed support sends $(a+Mv,v)$ to $a$, and $\wt(a)\le\wt(a+Mv)+\kappa\,\wt(v)$ when every kernel qubit lies in at most $\kappa$ sets. This proves the seed-support bounds of Sec.~\ref{sec:distance}. At $N=0$ the choice $v=0$ is allowed for every $M$ and gives $d_X\le d_X(\mathcal S)$, which for disjoint sets, $\kappa=1$, matches the seed-support bound. For a kernel with logical qubits the bound $d_X\le d_X(\mathcal S)$ still holds at $N=0$, but equality need not, since the kernel can introduce lighter logical operators. At $N\ne0$ that choice is allowed only when $Na\in K_Z^\perp$, so the twist can raise $d_X$. At $N=0$, Eq.~\eqref{eq:dzstate} becomes the $Z$ distance of Eq.~\eqref{eq:untwisted}.

If $K_Z=0$, the kernel is all-plus and $V_N$ acts trivially on the product code, so a twist can change the code only through the kernel's $Z$ stabilizers.

For fixed $M$ and $N$, enlarging $K_Z$ allows more $y$ in the $Z$ minimization, so $Z$ stabilizers of the kernel can only lower $d_Z$, as stated in Sec.~\ref{sec:distance}. It also allows fewer $v$ in the $X$ minimization, so they can only raise $d_X$.

\subsection{Transfer of diagonal seed gates}\label{app:ztransfer}
We prove Eq.~\eqref{eq:paritytransfer} for every kernel, with or without logical qubits, and for every $M$ and $N$. The circuit $U_MV_N$ permutes computational basis states, $\ket{s,u}\mapsto\ket{E(s,u)}$, and every image $(x,x')=E(s,u)$ satisfies $x+Mx'=s$. The gate $D_\oplus$ of Sec.~\ref{sec:nonclifford} therefore multiplies $\ket{E(s,u)}$ by $e^{i\phi(s)}$, the phase that $D\otimes I$ gives to $\ket{s,u}$. This proves Eq.~\eqref{eq:paritytransfer} on the full Hilbert space, not only on the code space. If $D$ is a logical gate of the seed, $D_\oplus$ therefore preserves $\mathcal Q$, acts on the seed's logical qubits as $D$ does, and acts as the identity on any logical qubits of the kernel. Statevector simulations of the Steane and 15-qubit Reed--Muller seeds, with random sets $F_i$ and twists, confirm this transfer for $T$ and $T^2$ with random kernels that have $Z$ stabilizers.

\subsection{Native faults and the untwisted equality}\label{app:nativefaults}
This appendix makes precise how many faulty native gates a layer tolerates, which Secs.~\ref{sec:nonclifford} and~\ref{sec:example} use, and then proves Eq.~\eqref{eq:untwistedblock}. A faulty native gate may apply an arbitrary error to its whole support, so a set of faulty gates is an arbitrary error on a set of parts of the partition $\mathcal B$ of Sec.~\ref{sec:nonclifford}. We first prove that arbitrary faults on any $s$ parts are always correctable exactly when $\Dpart\ge2s+1$. A product of two errors, each on at most $s$ parts, occupies at most $2s$ parts. If $2s<\Dpart$, that product is detectable or a stabilizer. The quantum error-correction conditions give ideal recovery for all such errors. Conversely, split a logical on the fewest parts into two pieces of at most $\lceil\Dpart/2\rceil$. They have equal syndrome and different logical action. No decoder corrects all errors of that size. The largest such $s$ is therefore $\tpart$ in
\begin{equation}
 \left\lceil d/\ell\right\rceil\le\Dpart\le d,\qquad
 \tpart=\left\lfloor\frac{\Dpart-1}{2}\right\rfloor ,
 \label{eq:locbound}
\end{equation}
where the bounds hold for a layer whose gates act on at most $\ell$ qubits. The lower bound holds because $j$ parts cover at most $j\ell$ qubits, while every nontrivial logical operator has weight at least $d$. The upper bound holds because a logical operator of weight $d$ meets at most $d$ parts. Every pair layer of the tesseract has $\Dpart=2$, since in each of them two pairs cover the support of a weight-four logical operator, while $H^{\otimes16}$ and $\mathrm{CNOT}^{\otimes16}$ have parts that hold one qubit of each block they act on, and therefore $\Dpart=d=4$.

A gate of width $\ell$ spreads one fault to at most $\ell$ qubits. The distance alone guarantees that the code detects such an error when $\ell\le d-1$ and corrects it when $\ell\le t=\lfloor(d-1)/2\rfloor$, that is, when $d\ge2\ell+1$. In that case $d/\ell>2$, so Eq.~\eqref{eq:locbound} gives $\Dpart\ge\lceil d/\ell\rceil\ge3$ for the partition of every generating layer. Each layer then has $\tpart\ge1$ and corrects one arbitrary faulty native gate, which is the condition of Table~\ref{tab:ft-codes}. A parity layer also obeys $\Dpart\le d_Z(\mathcal S)$ by Eq.~\eqref{eq:ceiling}. The codes of Table~\ref{tab:ft-codes} on the 15-qubit seed therefore have exactly $\Dpart=3$ on their parity layer, even the three $\qcode{59}{1}{9}$ codes \texttt{lift:width3-59-w8}, \texttt{lift:width3-59}, and \texttt{lift:twisted59}, of distance nine. With $\Dpart=3$, two faulty native gates in one layer can already cause a logical error. Their error can differ by a logical operator from an error on the third part of a three-part logical support, and ideal recovery can then apply the wrong correction. Such a layer fails at order $p^2$ in the native-gate error $p$, while the memory of a distance-nine code, with $t=4$, fails only at order $p^5$. More generally, a layer fails at order $p^{\tpart+1}$.

The same counting applies to constant-depth circuits, since every such circuit is also a depth-one layer, at a larger width. A circuit of $\Delta$ layers equals one layer whose gate supports are the connected components of the union of its gate supports. The merged native gate on each component is the restriction of the circuit to it. Its width $\ell_{\rm eff}$ is the size of the largest component. This width can reach $n$ already at $\Delta=2$, for example for two staggered layers of two-qubit gates along a line. The argument for $d\ge2\ell+1$ applies with $\ell_{\rm eff}$ in place of $\ell$, so one arbitrary faulty merged gate is correctable when $d\ge2\ell_{\rm eff}+1$. In general it is correctable exactly when the partition into components has $\Dpart\ge3$. A fault of one of the original gates of width $\ell$ is more benign, since it spreads only through its forward light cone of at most $\ell^\Delta$ qubits. Such faults are therefore correctable when $d\ge2\ell^\Delta+1$, even when the components are larger. The codes of Table~\ref{tab:ft-codes} need neither bound beyond $\Delta=1$, since each of their generators is a single layer with $d\ge2\ell+1$.

We then prove Eq.~\eqref{eq:untwistedblock}. At $N=0$ and for disjoint sets $F_i$, $U_M$ acts within the partition into these sets. Before coupling, a normalizer factors into seed and kernel normalizers, with at least one logical factor. The number of parts it occupies is at least that of either factor. Taking the other factor to be identity attains the corresponding minimum. Hence
\begin{equation}
 \Dpart(\mathcal Q;\mathcal B)=
 \min\{d(\mathcal S),\Dpart(\mathcal K;\mathcal B_K)\},
 \label{eq:locproduct}
\end{equation}
where $\mathcal B$ is the partition into the sets $F_i$ and singletons for the other kernel qubits, and $\mathcal B_K$ groups the kernel qubits of each $F_i$ and keeps every kernel qubit outside the sets $F_i$ as a singleton. A kernel without logical qubits contributes infinity to the minimum. This proves Eq.~\eqref{eq:untwistedblock}. The same argument holds at $N\ne0$ when every CNOT of $V_N$ stays inside one set $F_i$. Otherwise the exact part-counting minima of Eqs.~\eqref{eq:dxgeneral} and~\eqref{eq:dzgeneral} apply.

\subsection{Transfer of $X$-diagonal seed gates for the all-plus kernel}\label{app:xtransfer}
Gates diagonal in the $X$ basis do not transfer for every kernel. For the all-plus kernel they transfer through copies rather than parities, as stated in Sec.~\ref{sec:cliffordlayers}.
\begin{lemma}[Transfer of $X$-diagonal seed gates]\label{lem:xtransfer}
Let the kernel be all-plus, $\mathcal K=\ket+^{\otimes r}$, with any $M$ and $N$. Let $D_X$ be a logical gate of the seed that is diagonal in the $X$ basis, with phase $e^{i\phi(z)}$ on the $X$-basis string $z\in\F^m$ of the seed qubits. For each seed qubit $i$ choose a qubit $\sigma(i)\in F_i$ that lies in no other set, for example seed qubit $i$ itself. Then the operator that is diagonal in the $X$ basis, with phase $e^{i\phi(x_{\sigma(1)},\ldots,x_{\sigma(m)})}$ on the physical $X$-basis string $x$, preserves $\mathcal Q$ and acts on it as $D_X$ does on the seed. In particular, if $D_X=\prod_j\sqrt X(L_j)$ for sets $L_j$ of seed qubits, the same logical gate is implemented by $\prod_j\sqrt X(L'_j)$, where $L'_j$ replaces each seed qubit $i\in L_j$ by any qubit of $F_i$ that lies in no other set, chosen independently for each $j$.
\end{lemma}
\begin{proof}
A CNOT whose target is in $\ket+$ acts as the identity, so $V_N$ acts trivially on every state $\psi\otimes\ket+^{\otimes r}$, and $\mathcal Q$ equals the untwisted code. In the $X$ basis, $\CNOT(a\to b)$ adds the bit of $b$ to the bit of $a$. Each CNOT of $U_M$ has a kernel qubit as control and a seed qubit as target, so $U_M$ maps $\ket z_X\otimes\ket0_X$ to the $X$-basis string in which seed qubit $i$ carries $z_i$ and each kernel qubit carries the sum of $z_i$ over the sets $F_i$ that contain it. In particular, every qubit that lies only in $F_i$ carries $z_i$. Now let $D$ be any operator on the seed qubits that is diagonal in the $X$ basis, with phase $e^{i\phi(z)}$, logical or not, and let $\Phi_\sigma$ be the physical operator built from it as in the statement. On each of the strings above $x_{\sigma(i)}=z_i$, so $\Phi_\sigma U_M=U_M(D\otimes I)$ on the whole space of states $\psi\otimes\ket+^{\otimes r}$. For $D=D_X$, which preserves the seed code, the right side preserves $\mathcal Q=U_M(\mathcal S\otimes\ket+^{\otimes r})$ and acts on it as $D_X$ acts on the seed. This proves the first statement. Since $D\otimes I$ maps the space of states $\psi\otimes\ket+^{\otimes r}$ into itself, the relation composes under products, with an independent choice $\sigma_j$ for each factor. The operator $\sqrt X(L'_j)$ is $\Phi_{\sigma_j}$ for $D=\sqrt X(L_j)$, because its phase depends only on the parity of the bits of $L'_j$, although $\sqrt X(L_j)$ alone need not be a logical gate. Hence $\prod_j\sqrt X(L'_j)\,U_M=U_M(D_X\otimes I)$, and the same argument proves the second statement.
\end{proof}
The supports $L'_j$ can be chosen disjoint exactly when each seed qubit $i$ lies in at most as many of the sets $L_j$ as $F_i$ has qubits in no other set. For disjoint sets this number is $|F_i|$. A qubit shared by several sets carries the sum of their bits and cannot serve as a copy. For a kernel with $Z$ stabilizers, even the qubits of $F_i$ that lie in no other set no longer carry equal $X$-basis bits, and Eq.~\eqref{eq:constraint} selects choices that still work. For the $\qcode{37}{1}{7}$ construction we checked both statements on every $X$-basis string of the code. With the all-plus kernel, 200 random choices of copies for the seven Fano lines all implement $\sqrt{\bar X}^{\,\dagger}$ on all 2048 strings. With the GHZ kernel, none of 200 random choices does, while the choice of Eqs.~\eqref{eq:rootlines} and~\eqref{eq:ghz} in Appendix~\ref{app:37} does on all 524288 strings. The statevector simulations above also confirm the lemma.

\subsection{Kernels with $Z$ stabilizers and the kernel constraint}\label{app:constraint}
For a kernel without logical qubits other than the all-plus state, Eq.~\eqref{eq:constraint} is a sufficient condition for a layer of $\sqrt X(L)$ rotations to keep its action. In the $X$ basis, a transported state with a kernel without logical qubits has support
\begin{equation}
 E^{-T}\binom{z_0+h}{y},\qquad h\in S_Z,\quad y\in K_Z,
\end{equation}
where $z_0\in S_X^\perp$ labels the seed's logical $X$ state. Let $\Lambda$ be the binary matrix whose rows are the indicator vectors of the gate supports $L'_j$. The layer reads the gate parities $\Lambda E^{-T}(z_0+h,0)+\Lambda E^{-T}(0,y)$. If Eq.~\eqref{eq:constraint} holds, the kernel word changes no gate parity. The complete phase then equals that of the all-plus construction, which proves Eq.~\eqref{eq:constraint}, including logical phases. Writing $\Lambda=(\Lambda_S\ \Lambda_K)$ for its columns on seed and kernel qubits, the condition reads $K_Z\subseteq\ker[\Lambda_K+(\Lambda_S+\Lambda_KM^T)N^T]$, which is $K_Z\subseteq\ker\Lambda_K$ at $N=0$.

The criterion is sufficient but not necessary. For a $Z$ stabilizer of $\mathcal Q$, let $\mathcal A$ be the set of rotation generators $X(L'_j)$ that anticommute with it, the gates whose parity it flips. Conjugation by the layer maps the stabilizer to itself times $\prod_{j\in \mathcal A}\pm iX(L'_j)$, with the sign of each rotation. The layer therefore preserves $\mathcal Q$ exactly when this product is an $X$ stabilizer of $\mathcal Q$, phase included, for every $Z$ stabilizer. It suffices to check the kernel checks $WZ(y)W^\dagger$, for which Eq.~\eqref{eq:constraint} is the case $\mathcal A=\emptyset$. For a seed $Z$ stabilizer, the all-plus assumption places the product in the all-plus stabilizer group. It also commutes with every kernel check, because that is the same condition as the kernel product commuting with the seed stabilizer: the two sets $\mathcal A$ meet in an even number of gates. Preservation also fixes the logical action. The logical operators of $\mathcal Q$ can be chosen among those of the all-plus code, and a stabilizer of the all-plus code that commutes with every stabilizer of $\mathcal Q$ is a stabilizer of $\mathcal Q$. For example, couple the Steane code without twist to one kernel qubit per seed qubit, impose one kernel $Z$ check on the four copies of a seed $X$ check, and apply $\sqrt X$ to these four copies and to the other three seed qubits. The check anticommutes with four of the seven rotations, their factors $iX$ multiply to the $X$ stabilizer on the four copies, and the layer implements the same logical gate as transversal $\sqrt X$ on the seed.

\subsection{Converse for a physical parity layer}\label{app:converse}
We prove the converse stated at the start of Sec.~\ref{sec:construction}: every CSS code with a physical parity layer is a coupled code. The proof rests on a structure lemma that holds for every CSS code and every split of its qubits, with no gate involved. It is Goursat's lemma~\cite{goursat1889} for subspaces of a direct sum, together with the observation that over $\F$ the isomorphism it produces lifts to a twist.
\begin{lemma}[Structure of a CSS code across a cut]\label{lem:goursat}
Let $\mathcal Q=\CSS(Q_X,Q_Z)$ be a positive CSS code, split its qubits into a set $A$ of $m$ qubits and a set $B$ of $r$ qubits, and regard $C_X=Q_X$ and $C_Z=Q_Z^\perp$ as subspaces of $\F^m\oplus\F^r$. With the projection $\pi(s,u)=s$ onto $A$, define
\begin{align}
 S_X&=\pi C_X,\qquad S_Z=(\pi C_Z)^\perp,\nonumber\\
 K_X&=\{u:(0,u)\in C_X\},\qquad K_Z=\{u:(0,u)\in C_Z\}^\perp.
 \label{eq:goursatdata}
\end{align}
Then $S_X\subseteq S_Z^\perp$ and $K_X\subseteq K_Z^\perp$, and there is a twist $N\in\F^{r\times m}$, with CNOTs from the qubits of $A$ onto those of $B$ as in Eq.~\eqref{eq:coupling}, such that
\begin{equation}
 V_N\,\mathcal Q=\CSS(S_X,S_Z)\otimes\CSS(K_X,K_Z).
 \label{eq:goursatsplit}
\end{equation}
In particular $k=k_{\mathcal S}+k_{\mathcal K}$. The twist may be changed by any map that sends $S_X$ into $K_X$ and $S_Z^\perp$ into $K_Z^\perp$, and $N=0$ works exactly when $C_X=S_X\oplus K_X$ and $C_Z=S_Z^\perp\oplus K_Z^\perp$.
\end{lemma}
\begin{proof}
Goursat's lemma states that a subspace $C_X\subseteq\F^m\oplus\F^r$ is determined by its projections, intersections, and induced quotient isomorphism: $(s,u)\in C_X$ exactly when $s\in\pi C_X$ and $u$ lies in a coset of $K_X$ that depends linearly on $s$, and $\dim C_X=\dim\pi C_X+\dim K_X$. The same holds for $C_Z$, with $K_Z^\perp$ in place of $K_X$. Since $C_X\subseteq C_Z$, we have $\pi C_X\subseteq\pi C_Z$ and $K_X\subseteq K_Z^\perp$, which are the two inclusions. Choose a basis of $S_X$ and lift each vector to an element of $C_X$. Extend it to a basis of $S_Z^\perp$ and lift the new vectors to elements of $C_Z$. Let $N$ send each basis vector to the $B$ component of its lift, extended arbitrarily to $\F^m$. For $(s,u)\in C_X$, subtracting the lift of $s$ leaves an element $(0,u+Ns)$ of $C_X$, so $u+Ns\in K_X$, and likewise $u+Ns\in K_Z^\perp$ for $(s,u)\in C_Z$. The twist maps $X(s,u)$ to $X(s,u+Ns)$ and the basis string $(s,u)$ to $(s,u+Ns)$, so $V_NC_X=S_X\oplus K_X$ and $V_NC_Z=S_Z^\perp\oplus K_Z^\perp$. A CSS code is fixed by its $X$ checks and by the strings that occur in its codewords, which proves Eq.~\eqref{eq:goursatsplit}. The dimension formula for $C_X$ and $C_Z$ gives $k=\dim C_Z-\dim C_X=k_{\mathcal S}+k_{\mathcal K}$. Two twists give the same splitting exactly when their difference preserves the cosets above, and $N=0$ is such a twist exactly when $C_X$ and $C_Z$ are already direct sums.
\end{proof}
Now let $\mathcal Q$ carry a parity layer on sets $F_i$, each with a qubit that lies in no other set. Take that qubit as seed qubit $i$, let $A$ be the seed qubits and $B$ the remaining qubits, and let $U_M$ be the CNOTs of Eq.~\eqref{eq:coupling} from the kernel qubits of each $F_i$ onto seed qubit $i$. The code $U_M^{-1}\mathcal Q$ is again a positive CSS code, and Lemma~\ref{lem:goursat} gives $\mathcal Q=U_MV_N\bigl(\CSS(S_X,S_Z)\otimes\CSS(K_X,K_Z)\bigr)$. Undoing $U_M$ turns the parity of each $F_i$ into the bit of seed qubit $i$, so the layer becomes $D\otimes I$ with $D$ diagonal on the seed qubits. This operator commutes with $V_N$, whose CNOTs are controlled by seed qubits, and it preserves $V_N(\mathcal S\otimes\mathcal K)$ because the layer preserves $\mathcal Q$. Hence $D$ preserves $\mathcal S$. Its phase is constant on the cosets of $S_X$ in $S_Z^\perp$, and it is a logical seed gate. A non-Clifford physical logical image cannot arise from a Clifford seed image tensored with identity. Conjugating by $H^{\otimes n}$ exchanges $Q_X$ with $Q_Z$ and reverses every CNOT, so the same conclusion holds for a layer diagonal in the $X$ basis whose phase depends only on the $X$-basis parities of the sets, with the coupling CNOTs running from each seed qubit onto its kernel qubits.

The lemma uses no gate. The parity phase on the full physical Hilbert space enters only in identifying the layer as a seed gate, since undoing $U_M$ alone need not turn a layer into a product.

\section{Complete specification of the $\qcode{37}{1}{7}$ code}\label{app:37}
This appendix proves the properties of $\Qh$ stated in Sec.~\ref{sec:example}. Appendix~\ref{app:37seed} recalls the data that the proofs use, Appendix~\ref{app:37dist} proves $d=7$ and the optimal check weight $\Wt=8$, and Appendix~\ref{app:37gates} proves the action and the partition distance of each layer. The accompanying JSON file of Sec.~\ref{sec:families} lists the seed, the kernel, the coupling circuits and every generator, and Table~\ref{tab:schedule} lists every check. The proofs specialize the results of Appendix~\ref{app:theorems} to this code.

\subsection{Seed, kernel and coupling}\label{app:37seed}
The seed is the 15-qubit quantum Reed--Muller code~\cite{knill1996,bravyi2005}. We label its qubits $i=1,\ldots,15$ by the nonzero vectors of $\F^4$, the points of the projective space $\mathrm{PG}(3,2)$, identifying $i=i_1+2i_2+4i_3+8i_4$ with $(i_1,i_2,i_3,i_4)$. Three distinct points form a line when they sum to zero, so the lines are the 35 triples $\{i,j,i+j\}$. The $X$ stabilizers are $X(c_\mu)$ with $c_\mu=\{i:\mu\cdot i=1\}$ for $\mu\in\F^4$, so $S_X$ is the four-dimensional space of supports of linear functions. The $Z$ stabilizers are $Z(c)$ with $c\in S_Z=(S_X+\langle\mathbf1\rangle)^\perp$, where $\mathbf1$ is the set of all 15 qubits. The seven points with vanishing last coordinate form a Fano plane $\Pi=\{1,\ldots,7\}$, whose seven lines we call Fano lines, and the other eight points form a cube $O=\{8,\ldots,15\}$.

An operator $Z(z)$ commutes with the seed $X$ stabilizers exactly when the labels in $z$ sum to zero in $\F^4$. It is a stabilizer exactly when, in addition, $|z|$ is even. The minimum $Z$ logical is therefore a line, of weight three. The nontrivial $X$ coset consists of $\mathbf1$ and the complements of the sets $c_\mu$ with $\mu\ne0$, of weights 15 and seven, so $(d_X,d_Z)=(7,3)$. Every nonzero $g\in S_X$ has weight eight. The basis strings $a$ of the logical state $\ket{\bar\lambda}$, $\lambda\in\{0,1\}$, lie in $S_X$ for $\lambda=0$ and in $S_Z^\perp\setminus S_X$ for $\lambda=1$, so $\wt(a)\equiv7\lambda\pmod8$. Transversal $T$ multiplies $\ket a$ by $e^{i\pi\wt(a)/4}=e^{-i\pi\lambda/4}$ and therefore acts as $\bar T^\dagger$.

Physical qubits $1,\ldots,15$ are the seed qubits. Each seed qubit $i$ forms the set $F_i$ with two kernel qubits if $i\in\Pi$ and with one if $i\in O$,
\begin{equation}
 F_i=\begin{cases}
 \{i,14+2i,15+2i\},&1\le i\le7,\\
 \{i,i+22\},&8\le i\le15.
 \end{cases}\label{eq:37fibres}
\end{equation}
The other 22 qubits form the kernel. It holds GHZ states on the triples
\begin{equation}
 \{18,22,26\},\ \{19,24,28\},\ \{20,23,29\},\ \{21,25,27\},
 \label{eq:ghz}
\end{equation}
and $\ket+$ on its other ten qubits. Each GHZ triple holds one copy of each point of a Fano line not through point $1$. The kernel is thus $\CSS(K_Z^\perp,K_Z)$, with $K_Z$ spanned by the eight $ZZ$ checks of the GHZ triples. Each kernel qubit controls a CNOT onto the seed qubit of its set. These 22 CNOTs, of depth two, form the circuit $U_M$ of Sec.~\ref{sec:seedkernel} without a twist, so $WZ_iW^\dagger=Z(F_i)$, and they give the coupled code, \texttt{lift:ghz37} in Table~\ref{tab:ft-codes}.

A final circuit of 304 CNOTs brings the number of checks on each qubit, and hence the number of CNOTs that touch it in each extraction round, to at most six, while keeping both single-block layers, $d=7$, and the check weight eight. It is a product of 49 CNOT circuits of the following form. Each maps $X(v)\mapsto X(v+(h\cdot v)\,g)$ and $Z(v)\mapsto Z(v+(g\cdot v)\,h)$ for disjoint sets $g$ and $h$, where $h\cdot v=|h\cap v|\bmod2$. On $X$ supports it is the binary matrix $I+gh^{\mathsf T}$. Here $g$ has even overlap with every set $F_i$, and $h$ with each of the seven line triples of Appendix~\ref{app:37gates}. Each of them therefore fixes every $Z(F_i)$ and every $X(L)$ of a line triple. The two single-block layers are functions of these operators, so each of these circuits commutes with them on the full Hilbert space and keeps their logical action. The full circuit changes the stabilizer group, and the resulting code is $\Qh$. Conjugating the checks of the coupled code by the 304 CNOTs reproduces exactly the stabilizer spaces of the checks in Table~\ref{tab:schedule}, not only their dimensions, and all signs are positive. The 326 CNOTs define the code and do not constitute a verified fault-tolerant preparation protocol.

The seed supports of the four weight-eight $X$ checks form a basis of $S_X$, and every other $X$ check has empty seed support. Logical representatives are
\begin{equation}
 \bar X=X_{2,3,5,6,7,17,22},\qquad
 \bar Z=Z_{1,8,9,16,17,30,31},
 \label{eq:logicals}
\end{equation}
where $X_{a,b,\ldots}=X(\{a,b,\ldots\})$, and likewise for $Z$. Here $\bar Z$ is $Z$ on the sets $F_i$ of the line $\{1,8,9\}$, the image of the seed logical $Z_{1,8,9}$ by Eq.~\eqref{eq:zlogical}. The seed support of $\bar X$ is $\Pi$, the support of the seed logical $X(\Pi)$.

\subsection{Distance and optimal check weight}\label{app:37dist}
We first prove $d=7$ for the coupled code by minimizing Eq.~\eqref{eq:dzstate} of Appendix~\ref{app:statekernel} over the kernel $Z$ checks $y$. We then verify $d=7$ for $\Qh$ and prove $\Wt=8$. For a seed $Z$ logical $z$, the image $E^{-T}(z,0)$ has $Z$ on every qubit of the sets $F_i$ with $i\in z$. Each plane point of $z$ contributes its seed qubit, point $1$ adds the two qubits $16,17$ in $\ket+$, and each cube point adds its seed qubit and its copy in $\ket+$. The other copies lie in the GHZ triples, one triple for each Fano line not through point $1$. Adding an even $Z$ word $y$ on a triple reduces even support to zero and odd support to one qubit. Hence the minimum over $y$ in Eq.~\eqref{eq:dzstate} is the cost $\gamma(z)$ below, and $d_Z=\min_z\gamma(z)$ over the seed $Z$ logicals $z$, with
\begin{align}
 \gamma(z)={}&|z\cap\Pi|+2\mathbf1_{1\in z}+2|z\cap O|\nonumber\\
 &+\sum_{\substack{L\text{ a Fano line}\\1\notin L}}
       \bigl(|z\cap L|\bmod2\bigr),
 \label{eq:37cost}
\end{align}
where $\mathbf1_{1\in z}$ is one if $1\in z$ and zero otherwise.
The labels of $z$ sum to zero and the cube points are those with last coordinate one, so $|z\cap O|$ is even. The logical condition implies that $|z\cap\Pi|$ is odd. If $z$ contains cube points, they cost at least four. An odd plane subset costs at least three. A singleton at point one costs three directly, and any other singleton intersects two lines in the sum. Larger odd subsets already weigh at least three. Thus $\gamma(z)\ge7$ in this case.

If $z$ contains no cube points, it is an odd subset of $\Pi$ whose labels sum to zero, that is, one of the seven lines or the whole plane. A line not through point one costs seven. A line through point one costs nine, and the whole plane costs thirteen. The operator $\bar Z$ of Eq.~\eqref{eq:logicals} is a weight-seven $Z$ representative. This proves $d_Z=7$ for the coupled code. Equation~\eqref{eq:untwisted} gives $d_X=7$, with $X_{1\cdots7}$ as a witness.

The final CNOT circuit can change the weights of logical operators, so we determine the distance of $\Qh$ directly. In each Pauli sector we split every error of weight at most six into two halves of weight at most three and record their syndromes and logical parities. No two halves have equal syndrome and opposite logical parity, so no logical operator has weight at most six, and the two operators of Eq.~\eqref{eq:logicals} have weight seven. This proves $d_X=d_Z=7$.

Any full stabilizer generating set has an $X$ projection whose seed supports span $S_X$. At least one generator has nonzero seed support. Every nonzero vector of $S_X$ has weight eight. The sets $F_i$ are disjoint, so the seed support is no larger than the support, and that generator has physical weight at least eight. The final CNOT circuit leaves every seed support unchanged, since each $g$ has even overlap with every set $F_i$, so the argument applies to $\Qh$. This also allows mixed $X/Z$ generators. The checks of Table~\ref{tab:schedule} attain the bound, proving optimal $\Wt=8$.

For the all-plus kernel, $K_Z=0$, on the same sets $F_i$, a nonzero seed $Z$ check has at least four points, and a check that involves plane points weighs at least ten, because their number is even. The full $Z$ space needs such checks, and a weight-ten basis exists, so the all-plus code is $\code{37}{1}{7}{10}{3}$, listed as \texttt{uqlc:Fano37} in Table~\ref{tab:ft-codes}. This is the comparison of Sec.~\ref{sec:distance}, where the GHZ kernel lowers the optimal check weight from ten to eight.

\subsection{Gate generators and partition distances}\label{app:37gates}
The two single-block layers apply the parity rotations $T(L)$ and $\sqrt X(L)$ of Eq.~\eqref{eq:gates} on two partitions of the 37 qubits. On the partition into the sets $F_i$, the parity layer $\prod_{i=1}^{15}T(F_i)$ acts as $\bar T^\dagger$ by Eq.~\eqref{eq:paritytransfer} and the weight identity of Appendix~\ref{app:37seed}, and its square acts as $\bar S^\dagger$. On the line partition we apply $\sqrt X(L)$ on the three triples
\begin{equation}
 \{1,2,3\},\qquad \{16,4,5\},\qquad \{17,6,7\},
 \label{eq:rootlines}
\end{equation}
and on the four GHZ triples of Eq.~\eqref{eq:ghz}. Each of these line triples holds one copy of each point of a Fano line, and the other 16 qubits are idle singletons of the partition.

To see that the line layer implements $\sqrt{\bar X}^{\,\dagger}$, first take the all-plus kernel, in which every copy carries the $X$-basis bit $u_j$ of its point $j$. On a state with these bits, $\sqrt X(L)$ contributes the phase $e^{-i\pi/4}e^{i\pi x_L/2}$, where $x_L=\bigoplus_ju_j$ over the three points $j$ copied in $L$. Every point of $\Pi$ lies on three lines and every pair of points on one, so
\begin{equation}
 \sum_Lx_L\equiv3\bigoplus_{j\in\Pi}u_j\pmod4,
 \label{eq:fanoidentity}
\end{equation}
as the expansion of each parity as $\sum_ju_j-2\sum_{j<k}u_ju_k$ modulo four shows. Up to a global phase the layer therefore multiplies by $e^{-i\pi b/2}$, where $(-1)^b$ is the eigenvalue of $X(\Pi)$, a logical $\bar X$ of the coupled code. This is the action of $\sqrt{\bar X}^{\,\dagger}=e^{i\pi\bar X/4}$. The GHZ triples are chosen so that every GHZ $Z$ stabilizer has even overlap with every triple, and Eq.~\eqref{eq:constraint} then keeps this action for the GHZ kernel. The final CNOT circuit commutes with both layers, so $\Qh$ keeps both actions. Bitwise CNOT between corresponding qubits of two blocks implements $\overline{\CNOT}$, and its parts are these pairs of qubits.

Each faulty native gate is a fault on one part. Equation~\eqref{eq:locbound} gives $\Dpart\ge\lceil7/3\rceil=3$ for both single-block layers. The operator $\bar Z$ of Eq.~\eqref{eq:logicals} lies on three of the sets $F_i$, and the logical operator $X_{1,\ldots,7,16,17}$ lies on the first three triples of Eq.~\eqref{eq:rootlines}, so both layers have exactly $\Dpart=3$. Each part of bitwise CNOT holds one qubit of each block, and a nontrivial logical operator of two blocks has weight at least seven on one of them, so bitwise CNOT has $\Dpart=7$.

Three faulty native gates can apply an undetectable logical operator only in two kinds of patterns. In the parity layer these are the gates on the three sets $F_i$ of one of the 35 lines of $\mathrm{PG}(3,2)$, and in the line layer the three triples through one of the seven Fano points. A faulty triple applies one of 63 nonidentity Paulis and a faulty pair one of 15, and few of their products are logical. An exact count over all such patterns gives an undetectable logical operator with probability $5.3\times10^{-3}p^3$ in the parity layer and $1.6\times10^{-4}p^3$ in the line layer, when each native gate fails with probability $p$.

We checked $\Qh$ directly through an independent binary verifier. It confirmed both distances, all check ranks, the logical commutation relation, and both partition distances within one block. It also evaluated the phase of each layer on all $2^{18}=262144$ elements of each stabilizer group and of its logical coset: the parity layer gives phases $0$ and $7$ modulo $8$, and the line layer gives phases $0$ and $3$ modulo $4$, as $\bar T^\dagger$ and $\sqrt{\bar X}^{\,\dagger}$ require.

\section{Extraction circuits and noise model}\label{app:resources}
This appendix collects the inputs behind the logical error rates of Sec.~\ref{sec:example}. For $\Qh$ we give the extraction circuit and the proof of its circuit distance. We then define the noise model, the logical-error experiments and the exact simulation of the $\bar T^\dagger$ layer. Finally, we describe the post-selected preparation of the magic state quoted in Sec.~\ref{sec:conclusion}.

\subsection{Extraction circuit of the $\qcode{37}{1}{7}$ code}\label{app:extraction}
Because the stabilizer generators of $\Qh$ reach weight eight, bare-ancilla extraction would let a single ancilla fault spread to a high-weight data error, a so-called hook~\cite{dennis2002}. We measure each check with one ancilla and add flags~\cite{chamberland2018} only where they are needed. A flag is a qubit coupled to the ancilla by a pair of CNOTs that bracket part of the ancilla's CNOT sequence, so that a fault inside the bracket also flips the flag. The ancilla of a check together with its flags is a gadget, and only 14 of the 36 checks carry flags.

Table~\ref{tab:schedule} lists the round, and Fig.~\ref{fig:extraction} draws it. A $Z$ check prepares its ancilla $a$ in $\ket0$, applies $\CNOT(q\to a)$ for its data qubits $q$ in the listed order, and measures $a$ in $Z$. Each flag $f$ starts in $\ket+$, couples by $\CNOT(f\to a)$ at the two positions of its interval, and is measured in $X$. An $X$ check is the same circuit with the roles of $X$ and $Z$ exchanged. Its ancilla starts in $\ket+$ and applies $\CNOT(a\to q)$, its flags start in $\ket0$ and couple by $\CNOT(a\to f)$, and the bases of the readouts are exchanged. Measuring $X$ and $Z$ checks at the same time shortens the round and removes idle locations. An $X$ check and a $Z$ check that share qubits are then both measured correctly only if, on an even number of the shared qubits, the $X$ ancilla acts before the $Z$ ancilla. All 36 gadgets run in one schedule of 12 CNOT layers. Every ancilla and flag is prepared at the start of the round and measured at its end, so no qubit is reset or measured in the middle of a round. The round has $162$ data CNOTs and $2\times15$ flag CNOTs, 192 in all, on 88 qubits. The rows of Table~\ref{tab:schedule} are the 36 checks of $\Qh$, which we number $X1$--$X18$ and $Z1$--$Z18$ in the order of the table.

\begin{table*}[t]
\caption{The distance-seven extraction circuit of $\Qh$, with 15 flags, each used by one check. Data qubits are listed in CNOT order, and Wt is the weight of the check. A flag interval $(j,k)$ places one flag CNOT after the $j$-th and one after the $k$-th data CNOT, counting from one, with $j=0$ meaning before the first. Figure~\ref{fig:extraction} gives the layer of every CNOT.}\label{tab:schedule}
\centering\small\setlength{\tabcolsep}{4pt}
\begin{tabular}{lrl@{\hspace{14pt}}c@{\hspace{14pt}}lrl}\toprule
\multicolumn{3}{c}{$X$ checks} && \multicolumn{3}{c}{$Z$ checks}\\
\cmidrule(r){1-3}\cmidrule(l){5-7}
CNOT order & Wt & Flags && CNOT order & Wt & Flags\\\midrule
19, 18 & 2 &  && 3, 2 & 2 & \\
33, 11 & 2 &  && 16, 4 & 2 & \\
34, 12 & 2 &  && 6, 17 & 2 & \\
13, 35 & 2 &  && 23, 20 & 2 & \\
37, 15 & 2 &  && 20, 29 & 2 & \\
3, 21, 18, 2 & 4 &  && 9, 28, 24, 10 & 4 & (1,3)\\
22, 4, 1, 16 & 4 &  && 26, 5, 6, 28 & 4 & \\
17, 26, 6, 1 & 4 &  && 9, 32, 30, 14 & 4 & (1,3)\\
27, 26, 25, 5 & 4 &  && 27, 25, 31, 32 & 4 & (1,3)\\
7, 28, 5, 24 & 4 &  && 7, 25, 14, 23, 5, 8 & 6 & (1,5)\\
9, 31, 32, 10 & 4 &  && 27, 19, 5, 21, 7, 18 & 6 & \\
30, 8, 14, 36 & 4 &  && 7, 16, 25, 24, 27, 22 & 6 & (1,3)\\
7, 23, 21, 20, 22, 29 & 6 &  && 27, 22, 26, 21, 1, 2 & 6 & \\
9, 31, 25, 24, 36, 14 & 6 &  && 31, 11, 14, 8, 33, 10 & 6 & (1,5)\\
34, 15, 5, 6, 8, 33, 19, 17 & 8 & (1,7) && 26, 30, 1, 25, 36, 16 & 6 & (1,5)\\
3, 22, 12, 10, 2, 24, 33, 35 & 8 & (1,7) && 34, 9, 10, 12, 8, 36 & 6 & (1,5)\\
32, 7, 27, 18, 14, 15, 21, 11 & 8 & (1,6), (3,7) && 10, 13, 31, 36, 35, 30 & 6 & (1,5)\\
13, 9, 19, 26, 12, 21, 30, 28 & 8 & (1,7) && 32, 37, 31, 8, 15, 36 & 6 & (1,5)\\
\bottomrule\end{tabular}\end{table*}

We consider memory experiments of $R$ rounds, with fault locations after every CNOT, idle step, reset, and measurement. The circuit distance $d_{\rm circ}$ is the fewest faults that flip a logical observable without flipping any detector. For it we take the encoding and the final readout of all checks and of the logical to be ideal, and count the faults of the extraction rounds. A criterion on single gadgets does not by itself control $d_{\rm circ}$. Two unflagged hooks, each far from a logical operator, can combine with a hook from a third gadget and one data fault into a logical error. What does compose is the following condition. For a set $\mathcal E$ of ancilla faults in a gadget, let $n_{\rm f}(\mathcal E)$ be the number of flags raised and $w_{\rm red}(\mathcal E)$ the reduced weight of the resulting data error modulo the stabilizer group. If every gadget satisfies $|\mathcal E|+n_{\rm f}(\mathcal E)\ge w_{\rm red}(\mathcal E)$, a hook never costs fewer faults than the data error it is equivalent to, and therefore $d_{\rm circ}=d$. This condition is sufficient but costly in flags, so for $\Qh$ we instead compute the circuit distance of the whole round directly.

We prove $d_{\rm circ}\ge7$ by an exact enumeration in a conservative model. Each fault is described by the final syndrome, the logical parity, and the flags it raises, and the intermediate syndrome is discarded. Discarding it can only make distinct faults look alike, so the circuit distance in this model is a lower bound on the true one. We call this triple the signature of the fault. For $X$-type and $Z$-type faults separately, an exact enumeration of all sums of at most three signatures finds no two sums that differ only in logical parity. A set of at most six faults that flipped the logical without being detected would split into two sets of at most three faults with exactly such signatures. Hence no such set exists. Seven faults on data qubits suffice, so $d_{\rm circ}=d=7$. The result holds for any number of identical rounds between an ideal encoding and an ideal final readout of all checks and of the logical. The CNOT orders, flag positions, and layers were found by a constraint search. Its constraints are the order of each gadget's CNOTs, at most one CNOT per qubit and layer, and the ordering condition for every pair of overlapping $X$ and $Z$ checks. Finally, we propagate every single fault of the compiled round, at every CNOT, reset, measurement, and idle location, and check that its signature is among those used in the enumeration. A separate audit at every operation boundary checks that data errors do not spread to other data qubits or flags, and that a flag error changes only its own outcome. Noiseless simulation confirms that every detector is deterministic for one and three rounds.

Within this model the round is optimal in flags, CNOTs, CNOT layers, and qubits at once, for the checks of Table~\ref{tab:schedule}, one ancilla per check, and flags that are neither shared between checks nor reused. The weight-eight check $X17$ needs at least two flags, and thirteen other checks need at least one each, so at least 15 flags are required. The checks of Table~\ref{tab:schedule} have total weight 162, and each flag costs at least two CNOTs, so a round has at least $162+30=192$ CNOTs. The ancilla of $X17$ alone takes part in at least $8+4=12$ CNOTs, one per layer, so a round has at least 12 layers. The flag counts follow from a coloring argument. A fault on the $i$-th data CNOT of a check leaves the first $i-1$ or $i$ data qubits of its CNOT order as a data error, with the same flag outcomes. If such an error can be completed to an undetectable logical with at most $d-2$ further faults, its flag pattern must be nonzero. If the sum of two such errors can be completed with at most $d-3$ further faults, their flag patterns must differ. These conditions define a graph on the data CNOTs of the check, together with one vertex for the zero flag pattern, and $f$ flags can satisfy them only if the graph can be colored with $2^f$ colors. For every order of their data qubits, thirteen checks need at least one flag, and the eight-qubit check $X17$ gives a graph that is not bipartite, so it needs at least two. The other 22 checks need none. An independent enumeration of all orders confirms these bounds. 

\begin{figure*}[p]
\centering\includegraphics[width=\textwidth,height=0.84\textheight,keepaspectratio]{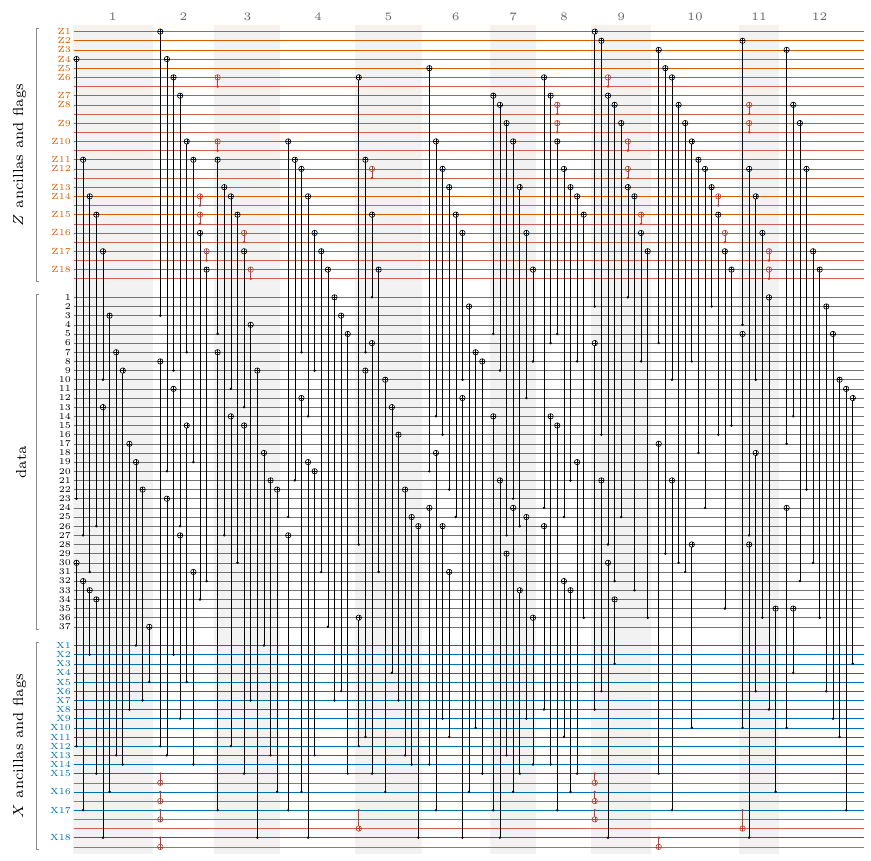}
\caption{The extraction round of $\Qh$ without flag reuse, with all 88 qubits (37 data, 36 syndrome ancillas, and 15 flags), all 192 CNOTs, and the 12 CNOT layers (numbered bands). Every ancilla and flag is prepared before the first layer and measured after the last one, so no flag is reused within a round. The data wires are in the middle, labeled from one, the $Z$-check ancillas (orange) with their flags (red) above, and the $X$-check ancillas (blue) with their flags below. Black CNOTs couple data to ancillas, and red CNOTs couple flags to ancillas. The round preserves the code distance, $d_{\rm circ}=7$, for any number of rounds between an ideal preparation and an ideal readout. The checks $Z1$--$Z18$ and $X1$--$X18$ are the rows of Table~\ref{tab:schedule}.}
\label{fig:extraction}
\end{figure*}

\subsection{Noise model and logical-error experiments}\label{app:noise}
Table~\ref{tab:noise} lists the noise model. All gate, reset, and measurement errors scale with the two-qubit gate error $p$, with the relative strengths of the neutral-atom model used throughout this work. A two-qubit fault chooses uniformly among the 15 nonidentity Paulis, and a faulty native triple among its 63. A reset or preparation fault prepares the opposite eigenstate, and a measurement fault flips the classical outcome. The simulated circuits contain no separate single-qubit gates, since each native gate is one location.

\begin{table}[t]
\caption{Noise model of the logical-error experiments. All probabilities are per operation. Idle errors follow Eq.~\eqref{eq:idle} with the durations and coherence times of Table~\ref{tab:coherence}.}
\label{tab:noise}\centering\small\setlength{\tabcolsep}{4pt}
\begin{tabular}{ll}\toprule
Location & Error\\\midrule
CNOT & two-qubit depolarizing, $p$\\
Native gate, $\ell\le3$ qubits & nonidentity Pauli, $p$\\
Ancilla and flag reset & wrong state, $p$\\
Ancilla and flag measurement & outcome flip, $5p/3$\\
Transversal data preparation & wrong state, $p$\\
Transversal data readout & outcome flip, $5p/3$\\
Idle qubit, time $t$ & Pauli channel of Eq.~\eqref{eq:idle}\\\bottomrule
\end{tabular}\end{table}

The hyperfine clock states that store the qubits relax with time $T_1$ and dephase with time $T_2$ under dynamical decoupling. Amplitude damping followed by pure dephasing, twirled over the Pauli group, is the Pauli channel that applies $X$, $Y$ and $Z$ with probabilities $p_x$, $p_y$ and $p_z$ to a qubit that waits for a time $t$,
\begin{equation}
 p_x=p_y=\frac{1-e^{-t/T_1}}{4},\qquad p_z=\frac{1-e^{-t/T_2}}{2}-\frac{1-e^{-t/T_1}}{4},
 \label{eq:idle}
\end{equation}
which is biased toward $Z$ when $T_2<T_1$. It acts on every qubit that waits, whether data, ancilla, or flag. A CNOT step consists of a Rydberg gate of about $0.2\,\mu$s and a parallel move of the ancillas to their next partners, and the move dominates. With constant acceleration $\alpha$, a move over a distance $\delta$ takes $2\sqrt{\delta/\alpha}$. For $\alpha=5500\,{\rm m\,s^{-2}}$ and a site spacing of $12\,\mu$m this is $93\,\mu$s for one site and $280\,\mu$s for nine sites~\cite{zhou2025}. We take $200\,\mu$s per step. While ancillas and flags are read out, the data qubits wait for $500\,\mu$s~\cite{zhou2025}. A native layer needs at most a one-site move, $100\,\mu$s, during which the qubits outside its supports wait. A transversal CNOT moves one block by its diameter, $400\,\mu$s, during which the data of both blocks wait.

\begin{table}[t]
\caption{Idle error from Eq.~\eqref{eq:idle} for $T_1=30$\,s and $T_2=10$\,s, a step of $200\,\mu$s, and a readout of $500\,\mu$s. A round counts the total idle error of an average data qubit, including the readout. It waits in about 8 of the 12 CNOT layers of $\Qh$.}
\label{tab:coherence}\centering\small\setlength{\tabcolsep}{3pt}
\begin{tabular}{lc}\toprule
Step, $(p_x,p_y,p_z)/10^{-5}$ & $(0.17,0.17,0.83)$\\
Readout & $2.9\times10^{-5}$\\
Round of $\Qh$ & $1.2\times10^{-4}$\\\bottomrule
\end{tabular}\end{table}

We take $T_1=30$\,s and $T_2=10$\,s, as assumed in Ref.~\cite{zhou2025}. These exceed the $T_1=4$\,s, including atom loss, and $T_2=1.5$\,s with XY16 decoupling measured in Ref.~\cite{bluvstein2022}, so the idle errors are optimistic. With the durations above, Eq.~\eqref{eq:idle} gives an idle error of $1.2\times10^{-5}$ per step and $2.9\times10^{-5}$ per readout, independent of $p$ (Table~\ref{tab:coherence}). We do not model atom loss as an erasure, heating or loss during moves, Rydberg crosstalk on spectators, leakage, or correlated readout errors.

The rates of memory and of the Clifford generators come from one experiment,
\begin{equation}
 \text{prep}\,-\,\text{SE round}\,-\,\bigl[\,\text{gate layer}\,-\,\text{SE round}\,\bigr]^{c}\,-\,\text{readout}.
 \label{eq:protocol}
\end{equation}
Here an SE round is one round of syndrome extraction, and a gate layer is one depth-one generating layer. The preparation puts every data qubit in $\ket0$, or every data qubit in $\ket+$, and the readout measures every data qubit in that basis. Memory is the case of an empty gate layer, and for each generator a cycle is one gate layer followed by one SE round. Since $\ket0^{\otimes n}$ is stabilized by every $Z$ check and by $\bar Z$, this transversal preparation needs no encoder. The first SE round projects it onto the code space, so its $Z$ detectors are deterministic, while its random $X$-check outcomes are the reference for the $X$ detectors of the second round. Preparing $\ket+^{\otimes n}$ exchanges the roles. A layer that does not keep the prepared basis, such as $\bar S^\dagger$ on $\ket{\bar+}$ or $\sqrt{\bar X}^{\,\dagger}$ on $\ket{\bar0}$, is run for an even number of cycles, so that its $c$-th power is a logical Pauli. For $\overline{\CNOT}$ both blocks are prepared in the same basis. The readout reconstructs the final checks of its basis and the logical observable. For each basis we run $c=2$, 4, and 8 cycles and fit the logical failure probability $P_L(c)$ after $c$ cycles as
\begin{equation}
 P_L(c)=\epsilon_0+c\,\epsilon,
 \label{eq:fit}
\end{equation}
by least squares weighted with the inverse Poisson variance of each failure count. Here $\epsilon_0$ collects the preparation, the first round, and the readout, and we report the sum of $\epsilon$ over the two bases as the failure per cycle.

Circuits are simulated with Stim~\cite{gidney2021stim}, which also converts each circuit into a detector error model with the probabilities of the noise model. Almost every fault mechanism of this model flips more than two detectors, because of the flags, the correlated faults of two- and three-qubit gates, and the interleaved round, so matching decoders do not apply directly. We decode every shot with a nontrivial syndrome with the Tesseract decoder~\cite{tesseract2025}, an A$^*$ search for the most likely set of fault mechanisms that explains the syndrome. Its search is pruned, so it is a heuristic minimum-weight decoder rather than a maximum-likelihood one.

The $\bar T^\dagger$ layer is simulated exactly with the Clifford-stabilizer method of Ref.~\cite{takada2026}. A Pauli error $P$ present at the layer becomes $PR(P)$, where $R(P)$ is the product of $S^\dagger(F_i)$ over the sets $F_i$ on which the $X$ part of $P$ has odd overlap, with $S(L)=T(L)^2$. Each shot samples every fault, propagates the faults to find the $X$ part of $P$ at each layer, and runs the stabilizer circuit without the layer, with the same faults, $R(P)$, and the controller's corrections inserted. The measurement record then has the distribution of the circuit with the layer. On an earlier extraction circuit of $\Qh$, an independent simulator that applies the $T$ gates themselves gives the same detector statistics, with $\chi^2$ per degree of freedom 1.03 over 138 detectors, against 5.4 when $R(P)$ is omitted. Statevector simulation of the 15-qubit seed with $\bar T$ and $\bar T^\dagger$ layers agrees within statistical error.

In the $X$ basis we prepare $\ket+^{\otimes n}$ and project the $Z$ checks noiselessly onto $+1$ before the first round, which fixes the Pauli frame that the controller would otherwise track. Before each layer the controller decodes the record so far with Tesseract and maps the decoded faults to the $X$ parities of the sets $F_i$. Right after the layer it applies $S(F_i)$ on the sets that it estimates odd. The residual $S(F_i)^{\pm1}$ on a misjudged set is projected by the next round onto the identity or onto $Z(F_i)$, except when several residuals combine into an operator with trivial syndrome, which the stabilizer simulation keeps exactly. After the $c$ cycles the $c$ layers are undone ideally in the controller's frame, and a noiseless round and readout of $\bar X$ end the experiment. We fit the failure linearly over $c=1,2,4,8$ and report the slope. In the $Z$ basis the layer is diagonal, and the readout measures $\bar Z$ directly after four noisy cycles. The final decoder adds a mechanism $Z(F_i)$ for every set and layer to the detector error model. It re-estimates the parity of every set at every layer from the complete record. On the sets where this estimate differs from the controller's recorded decision, it tries $Z(F_i)$ on every combination of up to six of them at no cost and keeps the most likely explanation.

A misjudged set adds $Z(F_i)$, a $Z$ error on two or three qubits, to a single $X$ fault, so two such sets can already cause a logical error. For two sets $F_i$ and $F_j$, the product $Z(F_i)Z(F_j)$ differs by the logical operator of Eq.~\eqref{eq:zlogical} from $Z(F_k)$ on the third point $k$ of their line, and the decoder prefers $Z(F_k)$ when it is lighter. The controller cannot see an $X$ fault after a qubit's last $Z$-check CNOT, and the last round before the layer has no later round to resolve its own faults. It therefore misjudges at least one set in about one layer in ten. The hindsight of the final decoder uses only information that a real device has. Without it, the decoder fails about ten times as often, and diagnostic runs show that the misjudged sets, rather than the noise of the native gates, dominate the remaining failure.

\subsection{Post-selected magic-state preparation}\label{app:magic}
This appendix gives the factory preparation of the magic state $\bar T^\dagger\ket{\bar+}$ of $\Qh$ quoted in Sec.~\ref{sec:conclusion}. It uses the noise model and the exact simulation of the $\bar T^\dagger$ layer of Appendix~\ref{app:noise}.

Magic-state cultivation~\cite{gidney2024cultivation} injects a $\ket T$ state into a small color code, checks it by repeatedly measuring a transversal Clifford operator of which $\ket T$ is an eigenstate, discards every attempt with a nontrivial detector, and grows the accepted state into a larger surface code. In $\Qh$ the magic state $\bar T^\dagger\ket{\bar+}$ is one layer away from the stabilizer state $\ket{\bar+}$. Post-selection therefore needs only the detectors of ordinary extraction rounds, and no logical check of the state. An accepted state reaches a data block of the same code through a transversal CNOT~\cite{huang2023homomorphic}.

We prepare $\ket+^{\otimes37}$ and run one round. Its random $Z$-check outcomes define a Pauli frame, and the layer applies $T^\dagger$ instead of $T$ on the sets $F_i$ of odd frame parity. We keep the state only if every detector is trivial, so no decoder is needed. The detectors include those of a noiseless round after the layer, where a real factory would run a noisy one. We score an accepted state by undoing the layer ideally and measuring $\bar X$.

The failure probability is a sum over sets of faults that pass the post-selection, each weighted by its probability and evaluated exactly with the Clifford-stabilizer method of Appendix~\ref{app:noise}, including the frame correction of the first round. The sum includes every such set of at most four faults and omits heavier sets, so the result is an estimate rather than a bound. The acceptance probability is sampled from $2\times10^5$ shots.

At $p=10^{-3}$ half of the attempts are accepted, and an accepted state fails with $5.8\times10^{-12}$. The failure scales as $p^3$, since three faulty native gates on the sets of a line of $\mathrm{PG}(3,2)$ can form an undetectable logical error (Appendix~\ref{app:37gates}). This value describes the state right after the layer, before any further round and before its transfer into a data block. A transfer into a different code, such as the surface code, would instead need a logical Bell measurement between the two codes. Cultivation reaches $2\times10^{-9}$ at $p=10^{-3}$~\cite{gidney2024cultivation}, but under uniform depolarizing noise and with an $S$-state proxy for the $T$ state, which overstates its fault distance~\cite{chan2026cultivation}. Our value instead simulates the non-Clifford layer exactly. The two values are therefore not a like-for-like comparison.

\end{document}

%% file: tab_ftcodes.tex
\begin{table*}[!t]
\caption{\label{tab:ft-codes}%
The \NFT{} depth-one universal codes whose generating set is fault tolerant in the strict sense: every native gate acts on at most $\ell\le t=\lfloor(d-1)/2\rfloor$ qubits, so a single faulty gate leaves a correctable error. $\ell$ is the largest number of qubits on which one gate of the depth-one universal generating set in the JSON file acts, Clifford and non-Clifford layers alike. Codes are ordered by $\ell$, then by distance, then by $k$, then by $n$; thick lines separate $\ell$ groups, each with a single $\ell$, and thin lines separate distances within a group. Each code is named by its database label, its parameters followed by its source and version, under which the accompanying database stores its generating circuits. Wt and Deg are the largest check weight and qubit degree of the optimized stabilizer generating set that the database stores for the code (\texttt{ldpc\_generators}). A dagger marks the two decomposable codes, each universal only through layers coupling its two blocks.}
\begin{ruledtabular}
\footnotesize
\setlength{\tabcolsep}{3pt}
\setlength{\aboverulesep}{0pt}\setlength{\belowrulesep}{0pt}
\renewcommand{\arraystretch}{1.25}
\begin{tabular}{clrrl@{\hspace{3em}}clrrl}
$\ell$ & Code & Wt & Deg & Seed & $\ell$ & Code & Wt & Deg & Seed \\
\hline
\raisebox{-0.625\normalbaselineskip}[0pt][0pt]{2} & \stepcounter{ftcodes}$\qcode{128}{35}{8}$\,\texttt{cons:hcube734} & 32 & 19 & $\qcode{64}{15}{4}$ &  & \stepcounter{ftcodes}$\qcode{128}{20}{8}$\,\texttt{uqlc:PCR128} & 32 & 12 & $\qcode{8}{3}{2}$ \\
\cmidrule[0.4pt]{7-10}
 & \stepcounter{ftcodes}$\qcode{256}{36}{8}$\,\texttt{uqlc:Tensor256}$^{\dagger}$ & 32 & 20 & $\qcode{64}{9}{4}$ &  & \stepcounter{ftcodes}$\qcode{59}{1}{9}$\,\texttt{lift:width3-59-w8} & 8 & 12 & $\qcode{15}{1}{3}$ \\
\cmidrule[1.1pt]{1-5}
 & \stepcounter{ftcodes}$\qcode{37}{1}{7}$\,\texttt{lift:ghz37} & 8 & 6 & $\qcode{15}{1}{3}$ &  & \stepcounter{ftcodes}$\qcode{59}{1}{9}$\,\texttt{lift:width3-59} & 10 & 10 & $\qcode{15}{1}{3}$ \\
 & \stepcounter{ftcodes}$\qcode{37}{1}{7}$\,\texttt{uqlc:Fano37} & 10 & 5 & $\qcode{15}{1}{3}$ &  & \stepcounter{ftcodes}$\qcode{180}{4}{9}$\,\texttt{lift:lowcheck180}$^{\dagger}$ & 12 & 11 & $\qcode{60}{4}{3}$ \\
\cmidrule[0.4pt]{7-10}
 & \stepcounter{ftcodes}$\qcode{37}{1}{7}$\,\texttt{lift:corr37} & 10 & 6 & $\qcode{15}{1}{3}$ &  & \stepcounter{ftcodes}$\qcode{320}{70}{10}$\,\texttt{uqlc:GFP320} & 39 & 26 & $\qcode{64}{15}{4}$ \\
\cmidrule[0.4pt]{7-10}
 & \stepcounter{ftcodes}$\qcode{38}{1}{7}$\,\texttt{lift:corr38} & 10 & 6 & $\qcode{15}{1}{3}$ & 3 & \stepcounter{ftcodes}$\qcode{383}{70}{11}$\,\texttt{uqlc:GFP383} & 53 & 27 & $\qcode{64}{15}{4}$ \\
\cmidrule[0.4pt]{7-10}
 & \stepcounter{ftcodes}$\qcode{39}{1}{7}$\,\texttt{lift:corr39} & 10 & 6 & $\qcode{15}{1}{3}$ &  & \stepcounter{ftcodes}$\qcode{516}{70}{12}$\,\texttt{uqlc:GFP516} & 67 & 28 & $\qcode{64}{15}{4}$ \\
\cmidrule[0.4pt]{7-10}
 & \stepcounter{ftcodes}$\qcode{40}{1}{7}$\,\texttt{lift:corr40} & 10 & 6 & $\qcode{15}{1}{3}$ &  & \stepcounter{ftcodes}$\qcode{705}{70}{13}$\,\texttt{uqlc:GFP705} & 88 & 28 & $\qcode{64}{15}{4}$ \\
\cmidrule[0.4pt]{7-10}
 & \stepcounter{ftcodes}$\qcode{41}{1}{7}$\,\texttt{uqlc:Fano41} & 11 & 5 & $\qcode{15}{1}{3}$ &  & \stepcounter{ftcodes}$\qcode{754}{70}{14}$\,\texttt{uqlc:GFP754} & 96 & 28 & $\qcode{64}{15}{4}$ \\
\cmidrule[0.4pt]{7-10}
 & \stepcounter{ftcodes}$\qcode{41}{1}{7}$\,\texttt{lift:corr41} & 11 & 6 & $\qcode{15}{1}{3}$ &  & \stepcounter{ftcodes}$\qcode{761}{70}{15}$\,\texttt{uqlc:GFP761} & 96 & 29 & $\qcode{64}{15}{4}$ \\
\cmidrule[0.4pt]{7-10}
 & \stepcounter{ftcodes}$\qcode{42}{1}{7}$\,\texttt{lift:corr42} & 11 & 6 & $\qcode{15}{1}{3}$ &  & \stepcounter{ftcodes}$\qcode{768}{70}{16}$\,\texttt{uqlc:GFP768} & 96 & 29 & $\qcode{64}{15}{4}$ \\
\cmidrule[1.1pt]{6-10}
 & \stepcounter{ftcodes}$\qcode{43}{1}{7}$\,\texttt{lift:corr43} & 11 & 6 & $\qcode{15}{1}{3}$ &  & \stepcounter{ftcodes}$\qcode{59}{1}{9}$\,\texttt{lift:twisted59} & 10 & 11 & $\qcode{15}{1}{3}$ \\
\cmidrule[0.4pt]{7-10}
\raisebox{-0.625\normalbaselineskip}[0pt][0pt]{3} & \stepcounter{ftcodes}$\qcode{44}{1}{7}$\,\texttt{lift:corr44} & 12 & 6 & $\qcode{15}{1}{3}$ & 4 & \stepcounter{ftcodes}$\qcode{256}{70}{16}$\,\texttt{cons:qrm8} & 32 & 25 & $\qcode{64}{15}{4}$ \\
 & \stepcounter{ftcodes}$\qcode{45}{1}{7}$\,\texttt{lift:corr45} & 12 & 6 & $\qcode{15}{1}{3}$ &  & \stepcounter{ftcodes}$\qcode{1024}{210}{16}$\,\texttt{cons:hcube1046} & 128 & 44 & $\qcode{512}{84}{8}$ \\
\cmidrule[1.1pt]{6-10}
 & \stepcounter{ftcodes}$\qcode{45}{1}{7}$\,\texttt{cons:xlock45} & 12 & 8 & $\qcode{15}{1}{3}$ & \raisebox{-0.625\normalbaselineskip}[0pt][0pt]{5} & \stepcounter{ftcodes}$\qcode{139}{1}{13}$\,\texttt{lift:affine139} & 18 & 6 & $\qcode{31}{1}{3}$ \\
\cmidrule[0.4pt]{7-10}
 & \stepcounter{ftcodes}$\qcode{63}{2}{7}$\,\texttt{uqlc:SH-RM63} & 10 & 7 & $\qcode{28}{2}{3}$ &  & \stepcounter{ftcodes}$\qcode{155}{1}{15}$\,\texttt{lift:affine155} & 20 & 6 & $\qcode{31}{1}{3}$ \\
\cmidrule[1.1pt]{6-10}
 & \stepcounter{ftcodes}$\qcode{64}{2}{7}$\,\texttt{uqlc:SH-RM64} & 10 & 7 & $\qcode{28}{2}{3}$ &  & \stepcounter{ftcodes}$\qcode{149}{1}{13}$\,\texttt{lift:doubled149} & 24 & 7 & $\qcode{49}{1}{5}$ \\
\cmidrule[0.4pt]{7-10}
 & \stepcounter{ftcodes}$\qcode{65}{2}{7}$\,\texttt{uqlc:SH-RM65} & 10 & 7 & $\qcode{28}{2}{3}$ &  & \stepcounter{ftcodes}$\qcode{186}{1}{15}$\,\texttt{lift:qrm186} & 24 & 6 & $\qcode{31}{1}{3}$ \\
\cmidrule[0.4pt]{7-10}
 & \stepcounter{ftcodes}$\qcode{66}{2}{7}$\,\texttt{uqlc:SH-RM66} & 10 & 6 & $\qcode{28}{2}{3}$ & \raisebox{-0.625\normalbaselineskip}[0pt][0pt]{6} & \stepcounter{ftcodes}$\qcode{181}{1}{17}$\,\texttt{lift:doubled181} & 24 & 7 & $\qcode{49}{1}{5}$ \\
\cmidrule[0.4pt]{7-10}
 & \stepcounter{ftcodes}$\qcode{67}{2}{7}$\,\texttt{uqlc:SH-RM67} & 11 & 6 & $\qcode{28}{2}{3}$ &  & \stepcounter{ftcodes}$\qcode{378}{1}{18}$\,\texttt{lift:proj378} & 32 & 17 & $\qcode{127}{1}{7}$ \\
\cmidrule[0.4pt]{7-10}
 & \stepcounter{ftcodes}$\qcode{68}{2}{7}$\,\texttt{uqlc:SH-RM68} & 12 & 6 & $\qcode{28}{2}{3}$ &  & \stepcounter{ftcodes}$\qcode{474}{1}{24}$\,\texttt{lift:proj474} & 32 & 16 & $\qcode{127}{1}{7}$ \\
\cmidrule[0.4pt]{7-10}
 & \stepcounter{ftcodes}$\qcode{69}{2}{7}$\,\texttt{uqlc:SH-RM69} & 13 & 6 & $\qcode{28}{2}{3}$ &  & \stepcounter{ftcodes}$\qcode{570}{1}{30}$\,\texttt{lift:proj570} & 36 & 15 & $\qcode{127}{1}{7}$ \\
\cmidrule[1.1pt]{6-10}
 & \stepcounter{ftcodes}$\qcode{70}{2}{7}$\,\texttt{uqlc:SH-RM70} & 14 & 6 & $\qcode{28}{2}{3}$ & 9 & \stepcounter{ftcodes}$\qcode{1369}{1}{49}$\,\texttt{lift:ghz37sq} & 56 & 15 & $\qcode{225}{1}{9}$ \\
 & \stepcounter{ftcodes}$\qcode{121}{20}{7}$\,\texttt{uqlc:PCR121} & 30 & 12 & $\qcode{8}{3}{2}$ &  &  &  &  &  \\
\end{tabular}
\end{ruledtabular}
\end{table*}

%% file: tab_pscodes.tex
\begin{table*}[!t]
\caption{\label{tab:ps-codes}%
The \NPS{} depth-one universal codes that are fault tolerant only with post-selection: every native gate acts on at most $\ell\le d-1$ qubits, so a single faulty gate is detected, but $\ell>t=\lfloor(d-1)/2\rfloor$, so it is not always corrected. Table~\ref{tab:ft-codes} lists the codes with $\ell\le t$. $\ell$ is the largest number of qubits on which one gate of the depth-one universal generating set in the JSON file acts, Clifford and non-Clifford layers alike. Codes are ordered by $\ell$, then by distance, then by $k$, then by $n$; thick lines separate $\ell$ groups, each with a single $\ell$, and thin lines separate distances within a group. Each code is named by its database label, its parameters followed by its source and version, under which the accompanying database stores its generating circuits. Wt and Deg are the largest check weight and qubit degree of the optimized stabilizer generating set that the database stores for the code (\texttt{ldpc\_generators}).}
\begin{ruledtabular}
\footnotesize
\setlength{\tabcolsep}{3pt}
\setlength{\aboverulesep}{0pt}\setlength{\belowrulesep}{0pt}
\renewcommand{\arraystretch}{1.25}
\begin{tabular}{clrrl@{\hspace{3em}}clrrl}
$\ell$ & Code & Wt & Deg & Seed & $\ell$ & Code & Wt & Deg & Seed \\
\hline
 & \stepcounter{pscodes}$\qcode{15}{4}{3}$\,\texttt{cons:blk45k12} & 8 & 8 & $\qcode{8}{3}{2}$ & 2 & \stepcounter{pscodes}$\qcode{32}{10}{4}$\,\texttt{cons:hcube523} & 16 & 7 & $\qcode{16}{4}{2}$ \\
\cmidrule[1.1pt]{6-10}
 & \stepcounter{pscodes}$\qcode{18}{4}{3}$\,\texttt{xz:ad496c3d} & 8 & 5 & $\qcode{8}{3}{2}$ &  & \stepcounter{pscodes}$\qcode{12}{2}{4}$\,\texttt{eczoo:9a9707b2} & 6 & 6 & $\qcode{8}{2}{2}$ \\
 & \stepcounter{pscodes}$\qcode{15}{6}{3}$\,\texttt{cons:blk45k18} & 8 & 7 & $\qcode{8}{3}{2}$ &  & \stepcounter{pscodes}$\qcode{18}{2}{4}$\,\texttt{xz:2444970a} & 6 & 5 & $\qcode{8}{2}{2}$ \\
 & \stepcounter{pscodes}$\qcode{18}{6}{3}$\,\texttt{xz:7e3340e4} & 8 & 5 & $\qcode{8}{3}{2}$ &  & \stepcounter{pscodes}$\qcode{20}{2}{4}$\,\texttt{xz:a173aa5b} & 6 & 5 & $\qcode{8}{2}{2}$ \\
 & \stepcounter{pscodes}$\qcode{18}{6}{3}$\,\texttt{xz:abbfb50a} & 8 & 5 & $\qcode{8}{3}{2}$ &  & \stepcounter{pscodes}$\qcode{20}{2}{4}$\,\texttt{xz:99344804} & 7 & 5 & $\qcode{8}{2}{2}$ \\
 & \stepcounter{pscodes}$\qcode{18}{6}{3}$\,\texttt{xz:d930205a} & 8 & 5 & $\qcode{8}{3}{2}$ &  & \stepcounter{pscodes}$\qcode{19}{3}{4}$\,\texttt{xz:da778af6} & 7 & 5 & $\qcode{8}{2}{2}$ \\
 & \stepcounter{pscodes}$\qcode{18}{6}{3}$\,\texttt{xz:e8af9c51} & 8 & 5 & $\qcode{8}{3}{2}$ &  & \stepcounter{pscodes}$\qcode{24}{3}{4}$\,\texttt{lift:cube3-30-30} & 8 & 5 & $\qcode{8}{3}{2}$ \\
 & \stepcounter{pscodes}$\qcode{20}{6}{3}$\,\texttt{xz:49b7a4ce} & 8 & 5 & $\qcode{8}{3}{2}$ &  & \stepcounter{pscodes}$\qcode{24}{3}{4}$\,\texttt{lift:cube3-27-27} & 8 & 6 & $\qcode{8}{3}{2}$ \\
 & \stepcounter{pscodes}$\qcode{20}{6}{3}$\,\texttt{xz:aacac550} & 8 & 5 & $\qcode{8}{3}{2}$ &  & \stepcounter{pscodes}$\qcode{24}{3}{4}$\,\texttt{lift:cube3-34-34} & 12 & 5 & $\qcode{8}{3}{2}$ \\
 & \stepcounter{pscodes}$\qcode{20}{6}{3}$\,\texttt{xz:fb77d694} & 8 & 5 & $\qcode{8}{3}{2}$ &  & \stepcounter{pscodes}$\qcode{24}{3}{4}$\,\texttt{lift:cube3-37-37} & 12 & 5 & $\qcode{8}{3}{2}$ \\
 & \stepcounter{pscodes}$\qcode{20}{6}{3}$\,\texttt{xz:ffd5fd9f} & 8 & 5 & $\qcode{8}{3}{2}$ & \raisebox{-0.625\normalbaselineskip}[0pt][0pt]{3} & \stepcounter{pscodes}$\qcode{24}{3}{4}$\,\texttt{lift:cube3-33-33} & 12 & 6 & $\qcode{8}{3}{2}$ \\
\cmidrule[0.4pt]{2-5}
 & \stepcounter{pscodes}$\qcode{16}{3}{4}$\,\texttt{lift:cube2even} & 8 & 5 & $\qcode{8}{3}{2}$ &  & \stepcounter{pscodes}$\qcode{20}{4}{4}$\,\texttt{xz:8ebf9cd8} & 8 & 5 & $\qcode{8}{2}{2}$ \\
\cmidrule[0.4pt]{7-10}
 & \stepcounter{pscodes}$\qcode{16}{3}{4}$\,\texttt{lift:cube2max} & 8 & 5 & $\qcode{8}{3}{2}$ &  & \stepcounter{pscodes}$\qcode{29}{1}{5}$\,\texttt{uqlc:Fano29} & 8 & 5 & $\qcode{15}{1}{3}$ \\
 & \stepcounter{pscodes}$\qcode{16}{3}{4}$\,\texttt{lift:cube2rm} & 8 & 5 & $\qcode{8}{3}{2}$ &  & \stepcounter{pscodes}$\qcode{37}{1}{5}$\,\texttt{lift:qrm37} & 10 & 5 & $\qcode{15}{1}{3}$ \\
 & \stepcounter{pscodes}$\qcode{24}{3}{4}$\,\texttt{lift:cube3-15-15} & 10 & 5 & $\qcode{8}{3}{2}$ &  & \stepcounter{pscodes}$\qcode{29}{2}{5}$\,\texttt{cons:sq29} & 8 & 8 & $\qcode{15}{1}{3}$ \\
\raisebox{-0.625\normalbaselineskip}[0pt][0pt]{2} & \stepcounter{pscodes}$\qcode{16}{4}{4}$\,\texttt{xz:a08427d8} & 8 & 4 & $\qcode{8}{3}{2}$ &  & \stepcounter{pscodes}$\qcode{24}{3}{5}$\,\texttt{cons:tri24} & 8 & 6 & $\qcode{8}{3}{2}$ \\
\cmidrule[0.4pt]{7-10}
 & \stepcounter{pscodes}$\qcode{18}{4}{4}$\,\texttt{xz:262f9424} & 8 & 5 & $\qcode{8}{3}{2}$ &  & \stepcounter{pscodes}$\qcode{30}{1}{6}$\,\texttt{cons:xlock30} & 8 & 5 & $\qcode{15}{1}{3}$ \\
 & \stepcounter{pscodes}$\qcode{20}{4}{4}$\,\texttt{xz:436a271d} & 8 & 6 & $\qcode{8}{3}{2}$ &  & \stepcounter{pscodes}$\qcode{36}{1}{6}$\,\texttt{uqlc:Fano36} & 9 & 8 & $\qcode{15}{1}{3}$ \\
 & \stepcounter{pscodes}$\qcode{20}{4}{4}$\,\texttt{xz:fcd2f54b} & 8 & 6 & $\qcode{8}{3}{2}$ &  & \stepcounter{pscodes}$\qcode{38}{1}{6}$\,\texttt{lift:qrm38} & 10 & 5 & $\qcode{15}{1}{3}$ \\
 & \stepcounter{pscodes}$\qcode{19}{5}{4}$\,\texttt{xz:55b6311e} & 8 & 6 & $\qcode{8}{3}{2}$ &  & \stepcounter{pscodes}$\qcode{30}{3}{6}$\,\texttt{cons:glue15} & 8 & 7 & $\qcode{15}{1}{3}$ \\
 & \stepcounter{pscodes}$\qcode{16}{6}{4}$\,\texttt{xz:3c180907} & 8 & 5 & $\qcode{8}{3}{2}$ &  & \stepcounter{pscodes}$\qcode{72}{20}{6}$\,\texttt{uqlc:PC-QRM72} & 16 & 10 & $\qcode{8}{3}{2}$ \\
\cmidrule[1.1pt]{6-10}
 & \stepcounter{pscodes}$\qcode{16}{6}{4}$\,\texttt{copycup:66fe2995} & 8 & 8 & $\qcode{8}{3}{2}$ &  & \stepcounter{pscodes}$\qcode{30}{2}{6}$\,\texttt{cons:q0} & 8 & 8 & $\qcode{15}{1}{3}$ \\
\cmidrule[0.4pt]{7-10}
 & \stepcounter{pscodes}$\qcode{16}{6}{4}$\,\texttt{ml:e1c9cb94} & 8 & 8 & $\qcode{8}{3}{2}$ & \raisebox{-0.625\normalbaselineskip}[0pt][0pt]{4} & \stepcounter{pscodes}$\qcode{60}{1}{7}$\,\texttt{lift:qrm60} & 16 & 5 & $\qcode{15}{1}{3}$ \\
\cmidrule[0.4pt]{7-10}
 & \stepcounter{pscodes}$\qcode{20}{6}{4}$\,\texttt{xz:4b34cfd4} & 8 & 5 & $\qcode{8}{3}{2}$ &  & \stepcounter{pscodes}$\qcode{64}{20}{8}$\,\texttt{cons:qrm6} & 16 & 12 & $\qcode{24}{8}{4}$ \\
 & \stepcounter{pscodes}$\qcode{20}{6}{4}$\,\texttt{xz:a1c422e0} & 8 & 5 & $\qcode{8}{3}{2}$ &  & \stepcounter{pscodes}$\qcode{256}{56}{8}$\,\texttt{cons:hcube835} & 64 & 21 & $\qcode{128}{21}{4}$ \\
\cmidrule[1.1pt]{6-10}
 & \stepcounter{pscodes}$\qcode{20}{6}{4}$\,\texttt{xz:f3cea31d} & 8 & 5 & $\qcode{8}{3}{2}$ &  & \stepcounter{pscodes}$\qcode{117}{1}{9}$\,\texttt{lift:doubled117} & 24 & 7 & $\qcode{49}{1}{5}$ \\
 & \stepcounter{pscodes}$\qcode{20}{6}{4}$\,\texttt{xz:58c5eac1} & 10 & 5 & $\qcode{8}{3}{2}$ & \raisebox{-0.625\normalbaselineskip}[0pt][0pt]{6} & \stepcounter{pscodes}$\qcode{90}{2}{9}$\,\texttt{lift:lowcheck90} & 12 & 11 & $\qcode{30}{2}{3}$ \\
 & \stepcounter{pscodes}$\qcode{24}{6}{4}$\,\texttt{lift:cube3-14-27} & 8 & 6 & $\qcode{8}{3}{2}$ &  & \stepcounter{pscodes}$\qcode{90}{2}{9}$\,\texttt{lift:width3-90-k2} & 12 & 12 & $\qcode{30}{2}{3}$ \\
\cmidrule[0.4pt]{7-10}
 & \stepcounter{pscodes}$\qcode{24}{6}{4}$\,\texttt{lift:cube3-30-35} & 12 & 5 & $\qcode{8}{3}{2}$ &  & \stepcounter{pscodes}$\qcode{282}{1}{12}$\,\texttt{lift:proj282} & 32 & 22 & $\qcode{127}{1}{7}$ \\
\cmidrule[1.1pt]{6-10}
 & \stepcounter{pscodes}$\qcode{24}{6}{4}$\,\texttt{lift:cube3-27-33} & 12 & 6 & $\qcode{8}{3}{2}$ & \raisebox{-0.625\normalbaselineskip}[0pt][0pt]{10} & \stepcounter{pscodes}$\qcode{251}{1}{11}$\,\texttt{lift:incidence251} & 32 & 22 & $\qcode{127}{1}{7}$ \\
\cmidrule[0.4pt]{7-10}
 & \stepcounter{pscodes}$\qcode{24}{8}{4}$\,\texttt{ml:78f97bd9} & 12 & 6 & $\qcode{8}{3}{2}$ &  & \stepcounter{pscodes}$\qcode{272}{6}{20}$\,\texttt{lift:tesseract17} & 40 & 10 & $\qcode{8}{3}{2}$ \\
 & \stepcounter{pscodes}$\qcode{24}{9}{4}$\,\texttt{lift:cube3-14-33} & 12 & 6 & $\qcode{8}{3}{2}$ &  &  &  &  &  \\
\end{tabular}
\end{ruledtabular}
\end{table*}